\documentclass[twocolumn,10pt]{article}
\usepackage[small,compact]{titlesec} 
\usepackage{pslatex}
\usepackage{setspace}
\usepackage[margin=2.0cm]{geometry} 
\usepackage[ruled,vlined]{algorithm2e}

\SetAlFnt{\small}
\SetAlCapFnt{\small}
\SetAlCapNameFnt{\small}
\SetAlCapHSkip{0pt}
\IncMargin{-\parindent}

\usepackage{epsfig}
\usepackage{amssymb}
\usepackage{amsmath}
\usepackage{amsfonts}
\usepackage{eepic}
\usepackage{graphpap}
\usepackage{arydshln}
\usepackage{bigstrut}
\usepackage{stmaryrd}
\newcommand{\db}{{\mathbf{db}}}
\newcommand{\block}{{\mathbf{b}}}

\newcommand{\rep}{{\mathbf{r}}}
\newcommand{\pw}{{\mathbf{w}}}

\newcommand{\all}[1]{{\mathsf{vars}}({#1})}
\newcommand{\key}[1]{{\mathsf{key}}({#1})}
\newcommand{\FD}[1]{{\mathcal{K}}({#1})}
\newcommand{\step}[1]{\stackrel{#1}{\smallfrown}}

\newcommand{\fd}[2]{{#1}\rightarrow{#2}}

\newcommand{\cqak}[1]{{\mathsf{CERTAINTY}}({#1})}
\newcommand{\prok}[1]{{\mathsf{PROBABILITY}}({#1})}
\newcommand{\attacksymbol}[1]{\stackrel{#1}{\rightsquigarrow}}
\newcommand{\nattacksymbol}[1]{\stackrel{#1}{\not\rightsquigarrow}}
\newcommand{\attacks}[3]{{#1}\attacksymbol{#3}{#2}}
\newcommand{\nattacks}[3]{{#1}\nattacksymbol{#3}{#2}}
\newcommand{\keycl}[2]{{#1}^{+,{#2}}}
\newcommand{\formula}[1]{({#1})}
\newcommand{\bigformula}[1]{\big({#1}\big)}

\newcommand{\signature}[2]{[{#1},{#2}]}
\newcommand{\tuple}[1]{\langle{#1}\rangle}
\newcommand{\att}[1]{{\mathbf{#1}}}
\newcommand{\con}[1]{{\mbox{`#1'}}}

\newcommand{\ax}{1}
\newcommand{\substitute}[3]{{#1}_{[{{#2}\mapsto{#3}}]}}
\newcommand{\sjarp}{\natural}
\newcommand{\cq}[1]{{\mathsf{C}}({#1})}
\newcommand{\acq}[1]{{\mathsf{AC}}({#1})}
\newcommand{\ACO}{\mbox{\bf AC}$^0$}
\newcommand{\sharpCQA}[1]{\mathsf{\sjarp CERTAINTY}({#1})}
\newcommand{\sharpP}{\mbox{\bf $\sjarp$P}}
\newcommand{\coNP}{\mbox{\bf coNP}}

\newcommand{\PTIME}{\mbox{\bf P}}
\newcommand{\FPTIME}{\mbox{\bf FP}}
\newcommand{\waar}{{\mathbf{true}}}
\newcommand{\onwaar}{{\mathbf{false}}}
\newcommand{\queryvars}[1]{\mathsf{vars}({#1})}
\newcommand{\sequencevars}[1]{\mathsf{vars}({#1})}
\newcommand{\card}[1]{|{#1}|}
\newcommand{\type}[1]{\mathsf{type}({#1})}
\newcommand{\probsymbol}{\mathsf{Pr}}
\newcommand{\prob}[1]{\probsymbol({#1})}
\newcommand{\probset}{[0,1]}
\newcommand{\pworlds}[1]{\mathsf{worlds}({#1})}

\newcommand{\calC}{{\mathcal{C}}}
\newcommand{\calP}{{\mathcal{P}}}
\newcommand{\calV}{{\mathcal{V}}}
\newcommand{\calD}{{\mathcal{D}}}
\newcommand{\calR}{{\mathcal{R}}}
\newcommand{\attacksymbolq}[1]{\rightsquigarrow}
\newcommand{\nattacksymbolq}[1]{\not\rightsquigarrow}
\newcommand{\attacksq}[3]{{#1}\attacksymbolq{#3}{#2}}
\newcommand{\nattacksq}[3]{{#1}\nattacksymbolq{#3}{#2}}
\newcommand{\skolem}[2]{\langle{#1},{#2}\rangle}
\newcommand{\triskolem}[3]{\langle{#1},{#2},{#3}\rangle}
\newcommand{\rv}[1]{\widehat{#1}}
\newcommand{\keyclsup}[2]{{#1}^{\boxplus,{#2}}}

\newcommand{\rest}{{\mathsf{rest}}}
\newcommand{\transformsymbol}{{\mathsf{map}}}
\newcommand{\transform}[1]{\transformsymbol({#1})}
\newcommand{\SE}[1]{\mbox{\sf R{#1}:\phantom{x}}}
\newcommand{\se}[1]{\mbox{\sf R{#1}}}
\newcommand{\kc}[2]{{\mathsf{block}}({#1},{#2})}
\newcommand{\clean}[1]{\llfloor{#1}\rrfloor}
\newtheorem{conjecture}{Conjecture}
\newtheorem{lemma}{Lemma}
\newtheorem{theorem}{Theorem}
\newtheorem{corollary}{Corollary}
\newtheorem{proposition}{Proposition}

\newtheorem{sublemma}{Sublemma}
\newenvironment{proof}[1][]{{\bf Proof #1}\hspace{0.5ex}}{\hfill$\Box$\newline}
\newenvironment{subproof}[1][]{{\bf Proof #1}}{\hfill$\dashv$\newline}
\newtheorem{defi}{Definition}
\newenvironment{definition}{\begin{defi}\rm}{\hfill$\lhd$\end{defi}}
\newtheorem{exa}{Example}
\newenvironment{example}{\begin{exa}\rm}{\hfill$\lhd$\end{exa}}
\newcommand{\myparagraph}[1]{\medskip{\bf{#1}}}
\newcommand{\myhfill}{\medskip\hfill}

\title{\bf Charting the Tractability Frontier of Certain Conjunctive Query Answering}
\author{Jef Wijsen\\Universit\'e de Mons, Belgium}
\date{\vspace{-1\baselineskip}}
\begin{document}
\maketitle
\begin{abstract}
An uncertain database is defined as a relational database in which primary keys need not be satisfied.
A repair (or possible world) of such database is obtained by selecting a maximal number of tuples without ever selecting two distinct tuples with the same primary key value.
For a Boolean query $q$, the decision problem $\cqak{q}$ takes as input an uncertain database $\db$ and asks whether $q$ is satisfied by every repair of $\db$.
Our main focus is on acyclic Boolean conjunctive queries without self-join.
Previous work~\cite{DBLP:journals/tods/Wijsen12} has introduced the notion of (directed) attack graph of such queries,
and has proved that $\cqak{q}$ is first-order expressible if and only if the attack graph of $q$ is acyclic.
The current paper investigates the boundary between tractability and intractability of $\cqak{q}$.
We first classify cycles in attack graphs as either weak or strong, and then prove among others the following.
If the attack graph of a query $q$ contains a strong cycle,
then $\cqak{q}$ is \coNP-complete.
If the attack graph of $q$ contains no strong cycle and every weak cycle of it is terminal (i.e., no edge leads from a vertex in the cycle to a vertex outside the cycle), then $\cqak{q}$ is in \PTIME.
We then partially address the only remaining open case, i.e., when the attack graph contains some nonterminal cycle and no strong cycle. 
Finally,  we establish a relationship between the complexities of $\cqak{q}$ and evaluating $q$ on probabilistic databases.
\end{abstract}

\section{Introduction}\label{sec:intro}

Primary key violations are a natural way for modeling uncertainty in the relational model.
If two distinct tuples have the same primary key value, then at least one of them must be mistaken, but we do not know which one.
This representation of uncertainty is also used in probabilistic databases, 
where each tuple is associated with a probability and distinct tuples with the same primary key value are disjoint probabilistic events~\cite[page~35]{DBLP:series/synthesis/2011Suciu}.

In this paper, the term {\em uncertain database\/} is used for databases with primary key constraints that need not be satisfied.
A repair (or possible world) of an uncertain database $\db$ is a maximal subset of $\db$ that satisfies all primary key constraints.
Semantics of querying follows the conventional paradigm of {\em consistent query answering\/}~\cite{ARENAS99,DBLP:series/synthesis/2011Bertossi}:
Given a Boolean query $q$, the decision problem $\cqak{q}$ takes as input an uncertain database $\db$ and asks whether $q$ is satisfied by every repair of $\db$.
Notice that $q$ is not part of the input, so the complexity of the problem is data complexity.
The restriction to Boolean queries simplifies the technical treatment, but is not fundamental.

\begin{figure}
\begin{tabular}{c|*3{l}}
$\att{C}$ 
& $\underline{\att{conf}}$ & $\underline{\att{year}}$ & $\att{city}$\\[0.5ex]\cline{2-4}
& PODS & 2016 & Rome\bigstrut\\
& PODS & 2016 & Paris\\[0.4ex]\cdashline{2-4}
& KDD  & 2017 & Rome\bigstrut\\
\end{tabular}
\begin{tabular}{c|lc}
$\att{R}$ 
& $\underline{\att{conf}}$ & $\att{rank}$\\[0.4ex]\cline{2-3}
& PODS & A\bigstrut\\[0.5ex]\cdashline{2-3}
& KDD  & A\bigstrut\\
& KDD  & B
\end{tabular}
\caption{Uncertain database.}\label{fig:planning}
\end{figure}

Primary keys are underlined in the conference planning database of Fig.~\ref{fig:planning}.
Maximal sets of tuples that agree on their primary key, called {\em blocks\/}, are separated by dashed lines.
There is uncertainty about the city of PODS~2016, and about the rank of KDD.
The database has four repairs.
The query $\exists x\exists y\formula{\att{C}(\underline{x,y},\con{Rome})\land\att{R}(\underline{x},\con{A})}$ (Will Rome host some A conference?) is true in only three repairs.

The problem $\cqak{q}$ is in \coNP\ for first-order queries $q$ (a ``no" certificate is a repair falsifying $q$).
Its complexity for conjunctive queries has attracted the attention of several authors, also outside the database community~\cite{DBLP:conf/dlog/Bienvenu12}.
A major research objective is to find an effective method that takes as input a conjunctive query $q$ and decides to which complexity classes $\cqak{q}$ belongs, or does not belong.
Complexity classes of interest are the class of first-order expressible problems (or \ACO), \PTIME, and \coNP-complete. 

Unless specified otherwise, whenever we say ``query" in the remainder of this paper,
we mean a Boolean conjunctive query without self-join (i.e., without repeated relation names).
Such queries are called acyclic if they have a join tree~\cite{BeeriFMY83}.

Our previous work~\cite{DBLP:conf/pods/Wijsen10,DBLP:journals/tods/Wijsen12} has revealed the frontier between first-order expressibility and inexpressibility of $\cqak{q}$ for acyclic queries $q$.
In the current work, we study the frontier between tractability and intractability of $\cqak{q}$ for the same class of queries.
That is, we aim at an effective method that takes as input a query $q$ and decides whether $\cqak{q}$ is in \PTIME\ or \coNP-complete (or neither of the two, which is theoretically possible if \PTIME$\neq$\coNP~\cite{DBLP:journals/jacm/Ladner75}).
For queries with exactly two atoms, such a method was recently found by Kolaitis and Pema~\cite{DBLP:journals/ipl/KolaitisP12}, but moving from two to more than two atoms is a major challenge.

Uncertain databases become probabilistic by assuming that the probabilities of all repairs are equal and sum up to~$1$.
In probabilistic terms, distinct tuples of the same block represent disjoint (i.e., exclusive) events,
while tuples of distinct blocks are independent. 
Such probabilistic databases have been called {\em block-independent-disjoint\/} (BID).
The trac\-ta\-bility/intrac\-ta\-bility frontier of query evaluation on BID probabilistic databases has been revealed by Dalvi et al.~\cite{DBLP:journals/jcss/DalviRS11}.
Here, evaluating a Boolean query is a function problem that takes as input a BID probabilistic database and asks the probability (a real number between $0$ and $1$) that $q$ is true.
The decision problem $\cqak{q}$, on the other hand, simply asks whether this probability is equal to $1$.

In previous work~\cite{DBLP:journals/tods/Wijsen12}, we introduced the (directed) attack graph of an acyclic query, and showed that $\cqak{q}$ is first-order expressible if and only if $q$'s attack graph is acyclic.
In the current paper, we study attack graphs in more depth.
We will classify cycles in attack graphs as either weak or strong.
The main contributions can then be summarized as follows.
\begin{enumerate}
\item
If the attack graph of an acyclic query $q$ contains a strong cycle, then $\cqak{q}$ is \coNP-complete.
This will be Theorem~\ref{the:strongcycle}.
\item
If the attack graph of an acyclic query $q$ contains no strong cycle and all weak cycles of it are terminal (i.e., no edge leads from a vertex in the cycle to a vertex outside the cycle), then $\cqak{q}$ is in \PTIME.
This will be Theorem~\ref{the:outdegree}.
\item
The only acyclic queries $q$ not covered by the two preceding results have an attack graph with some nonterminal cycle and without strong cycle.
We provide supporting evidence for our conjecture that $\cqak{q}$ is tractable for such queries.
Our results imply that $\cqak{q}$ is tractable for ``cycle" queries $q$ of the form
$\exists^{*}(R_{1}\big(\underline{x_{1}},x_{2})$ $\land R_{2}(\underline{x_{2}},x_{3})\dots$ $\land R_{k-1}(\underline{x_{k-1}},x_{k})$
$\land R_{k}(\underline{x_{k}},x_{1})\big)$.
These queries arise in the work of Fuxman and Miller~\cite{FuxmanM07}.
The case $k=2$ was solved in~\cite{Wijsen2010950}, but the case $k>2$ was open and will be settled by Corollary~\ref{cor:fuxman}.
\item
Theorem~\ref{the:prob} and its Corollary~\ref{cor:prob} will establish a relationship between the tractability frontiers of $\cqak{q}$ and query evaluation on probabilistic databases.
\end{enumerate}
Our work significantly extends and generalizes known results in the literature.

The remainder of this paper is organized as follows.
The next section further discusses related work.
Section~\ref{sec:preliminaries} defines the basic notions of certain conjunctive query answering.
Section~\ref{sec:attackgraph} defines the notion of attack graph.
Sections~\ref{sec:intractability} and~\ref{sec:tractability} show our main intractability and tractability results respectively.  
Section~\ref{sec:bid} establishes a relationship between the complexities of $\cqak{q}$ and  evaluating query $q$ on probabilistic databases.
Section~\ref{sec:discussion} concludes the paper and raises challenges for future research.
Several proofs have been moved to an Appendix.

\section{More Related Work}\label{sec:related}

The investigation of $\cqak{q}$ was pioneered by Fuxman and Miller~\cite{FUXMAN2005,FuxmanM07}, who defined a class of queries $q$ for which $\cqak{q}$ is first-order expressible.
This class has later on been extended by Wijsen~\cite{DBLP:conf/pods/Wijsen10,DBLP:journals/tods/Wijsen12}, who developed an effective method to decide whether $\cqak{q}$ is first-order expressible for acyclic queries $q$.
In their conclusion, Fuxman and Miller~\cite{FUXMAN2005,FuxmanM07} raised the question whether there exist queries $q$, without self-join, such that $\cqak{q}$ is in \PTIME\ but not first-order expressible.
The first example of such a query was identified by Wijsen~\cite{Wijsen2010950}.
The current paper identifies a large class of such queries (all acyclic queries with a cyclic attack graph in which all cycles are weak and terminal). 

Kolaitis and Pema~\cite{DBLP:journals/ipl/KolaitisP12} recently showed that for every query $q$ with exactly two atoms, $\cqak{q}$ is either in \PTIME\ or \coNP-complete, and it is decidable which of the two is the case. 
If $\cqak{q}$ is in \PTIME\ and not first-order expressible, then it can be reduced in polynomial time to the problem of finding maximal (with respect to cardinality) independent sets of vertices in claw-free graphs.
The latter problem can be solved in polynomial time by an ingenious algorithm of Minty~\cite{DBLP:journals/jct/Minty80}.
Unfortunately, the proposed reduction is not applicable on queries with more than two atoms.

The counting variant of $\cqak{q}$, which has been denoted $\sharpCQA{q}$, takes as input an uncertain database $\db$ and asks to determine the number of repairs of $\db$ that satisfy query $q$.
Maslowski and Wijsen~\cite{DBLP:conf/lid/MaslowskiW11,MASLOWSKIJCSS2012} have  recently showed that for every query $q$, the counting problem $\sharpCQA{q}$ is either in \FPTIME\ or \sharpP-complete, and it is decidable which of the two is the case. 

As observed in Section~\ref{sec:intro},
uncertain databases are a restricted case of block-independent-disjoint (BID) probabilistic databases~\cite{DalviRS09,DBLP:journals/jcss/DalviRS11}.
This observation will be elaborated in Section~\ref{sec:bid}.

All aforementioned results assume queries without self-join.
For queries $q$ with self-joins, only fragmentary results about the complexity of $\cqak{q}$ are known~\cite{MARCINKOWSKI02,WijsenIS09}.
The extension to unions of conjunctive queries has been studied in~\cite{GRIECO05}. 

\section{Preliminaries}\label{sec:preliminaries}

We assume disjoint sets of {\em variables\/} and {\em constants\/}.
If $\vec{x}$ is a sequence containing variables and constants, then $\sequencevars{\vec{x}}$ denotes the set of variables that occur in $\vec{x}$, and $\card{\vec{x}}$ denotes the length of $\vec{x}$.

Let $U$ be a set of variables.
A {\em valuation over $U$\/} is a total mapping $\theta$ from $U$ to the set of constants.
Such valuation $\theta$ is extended to be the identity on constants and on variables not in $U$.

\myparagraph{Atoms and key-equal facts.}
Every {\em relation name\/} $R$ has a fixed {\em signature\/}, which is a pair $\signature{n}{k}$ with $n\geq k\geq 1$: the integer $n$ is the {\em arity\/} of the relation name and $\{1,2,\dots,k\}$ is the {\em primary key\/}. 
The relation name $R$ is {\em all-key\/} if $n=k$.
If $R$ is a relation name with signature $\signature{n}{k}$, then $R(s_{1},\dots,s_{n})$ is an {\em $R$-atom\/} (or simply atom), where each $s_{i}$ is either a constant or a variable ($1\leq i\leq n$).
Such atom is commonly written as $R(\underline{\vec{x}},\vec{y})$ where the primary key value $\vec{x}=s_{1},\dots,s_{k}$ is underlined and $\vec{y}=s_{k+1},\dots,s_{n}$.
A {\em fact\/} is an atom in which no variable occurs.
Two facts $R_{1}(\underline{\vec{a}_{1}},\vec{b}_{1}),R_{2}(\underline{\vec{a}_{2}},\vec{b}_{2})$ are {\em key-equal\/} if $R_{1}=R_{2}$ and $\vec{a}_{1}=\vec{a}_{2}$.

We will use letters $F,G,H,I$ for atoms, and $A,B,C$ for facts of an uncertain database.
For atom $F=R(\underline{\vec{x}},\vec{y})$, we denote by $\key{F}$ the set of variables that occur in $\vec{x}$,
and by $\all{F}$ the set of variables that occur in $F$, that is, $\key{F}=\sequencevars{\vec{x}}$ and $\all{F}=\sequencevars{\vec{x}}\cup\sequencevars{\vec{y}}$.

\myparagraph{Uncertain database, blocks, and repairs.}
A {\em database schema\/} is a finite set of {\em relation names\/}.
All constructs that follow are defined relative to a fixed database schema.

An {\em uncertain database\/} is a finite set $\db$ of facts using only the relation names of the schema.
A {\em block\/} of $\db$ is a maximal set of key-equal facts of $\db$.
If $A\in\db$, then $\kc{A}{\db}$ denotes the block of $\db$ containing $A$.
An uncertain database $\db$ is {\em consistent\/} if it does not contain two distinct facts that are key-equal 
(i.e., if every block of $\db$ is a singleton).
A {\em repair\/} of $\db$ is a maximal consistent subset of $\db$.\footnote{It makes no difference whether the word ``maximal" refers to cardinality of sets or set-containment.}

\myparagraph{Boolean conjunctive query.}
A {\em Boolean conjunctive query\/} is a finite set 
$q=\{R_{1}(\underline{\vec{x}_{1}},\vec{y}_{1})$, $\dots$, $R_{n}(\underline{\vec{x}_{n}},\vec{y}_{n})\}$ of atoms.
By $\queryvars{q}$, we denote the set of variables that occur in $q$.
The set $q$ represents the first-order sentence 
$$\exists u_{1}\dots\exists u_{k}\bigformula{R_{1}(\underline{\vec{x}_{1}},\vec{y}_{1})\land\dots\land R_{n}(\underline{\vec{x}_{n}},\vec{y}_{n})},$$ where $\{u_{1}, \dots, u_{k}\}=\queryvars{q}$.
The query $q$ is {\em satisfied\/} by uncertain database $\db$, denoted $\db\models q$, if there exists a valuation $\theta$ over $\queryvars{q}$ such that for each $i\in\{1,\dots,n\}$,
$R_{i}(\underline{\theta(\vec{x}_{i})},\theta(\vec{y}_{i}))\in\db$.
We say that $q$ has a {\em self-join\/} if some relation name occurs more than once in $q$ (i.e., if $R_{i}=R_{j}$ for some $1\leq i<j\leq n$).

The restriction to Boolean queries simplifies the technical treatment, but is not fundamental.
Since every relation name has a fixed signature, relevant primary key constraints are implicitly present in all queries;
moreover, primary keys will be underlined.

\myparagraph{Join tree and acyclic conjunctive query.}
The notions of join tree and acyclicity~\cite{BeeriFMY83} are recalled next.
A {\em join tree\/} for a conjunctive query $q$ is an undirected tree whose vertices are the atoms of $q$ such that the following condition is satisfied:
\begin{quote}
{\em Connectedness Condition.\/} Whenever the same variable $x$ occurs in two atoms $F$ and $G$, then 
$x$ occurs in each atom on the unique path linking $F$ and $G$.
\end{quote} 
Commonly, an edge between atoms $F$ and $G$ is labeled by the (possibly empty) set $\all{F}\cap\all{G}$.
The term {\em Connectedness Condition\/} appears in~\cite{DBLP:journals/jcss/GottlobLS02} and refers to the fact that the set of vertices in which $x$ occurs induces a connected subtree.
A conjunctive query $q$ is {\em acyclic\/} if it has a join tree.
The symbol $\tau$ will be used for join trees.
We write $F\step{L}G$ to denote an edge between $F$ and $G$ with label $L$.
A join tree is shown in Fig.~\ref{fig:ax} (left).

\myparagraph{Certain query answering.}
Given a Boolean conjunctive query $q$, $\cqak{q}$ is (the complexity of) the following set.
$$
\setlength{\arraycolsep}{0pt}
\begin{array}{ll}
\cqak{q}=\{\db\mid\ & \mbox{$\db$ is an uncertain database such}\\
                    & \mbox{that every repair of $\db$ satisfies $q$}\}
\end{array}
$$
$\cqak{q}$ is said to be {\em first-order expressible\/} if there exists a first-order sentence $\varphi$ such that for every uncertain database $\db$, $\db\in\cqak{q}$ if and only if $\db\models\varphi$.
The formula $\varphi$, if it exists, is called a {\em certain first-order rewriting of\/} $q$.

\myparagraph{Purified uncertain databases.}
Let $q$ be a Boolean conjunctive query.
An uncertain database $\db$ is said to be {\em purified relative to $q$\/} if for every fact $A\in\db$,
there exists a valuation $\theta$ over $\queryvars{q}$ such that $A\in\theta(q)\subseteq\db$.
Intuitively, every fact in a purified uncertain database is relevant for the query.
This notion of purified database is new and illustrated next.

\begin{example}
The uncertain database $\{R(\underline{a},b)$, $S(\underline{b},a)$, $S(\underline{b},c)\}$ is not purified relative to query $\{R(\underline{x},y), S(\underline{y},x)\}$ because it contains no $R$-fact that ``joins" with $S(\underline{b},c)$.
\end{example}

The following lemma implies that in the study of tractability of $\cqak{q}$,
we can assume without loss of generality that uncertain databases are purified;
this assumption will simplify the technical treatment.
Notice that the query $q$ in the lemma's statement is not required to be acyclic.

\begin{lemma}\label{lem:purified}
Let $q$ be a Boolean conjunctive query.
Let $\db_{0}$ be an uncertain database.
It is possible to compute in polynomial time an uncertain database $\db$ that is purified relative to $q$ such that

\centerline{$\db\in\cqak{q}\iff\db_{0}\in\cqak{q}$.}
\end{lemma}

\section{Attack Graph}\label{sec:attackgraph}

\begin{figure*}\centering
\includegraphics{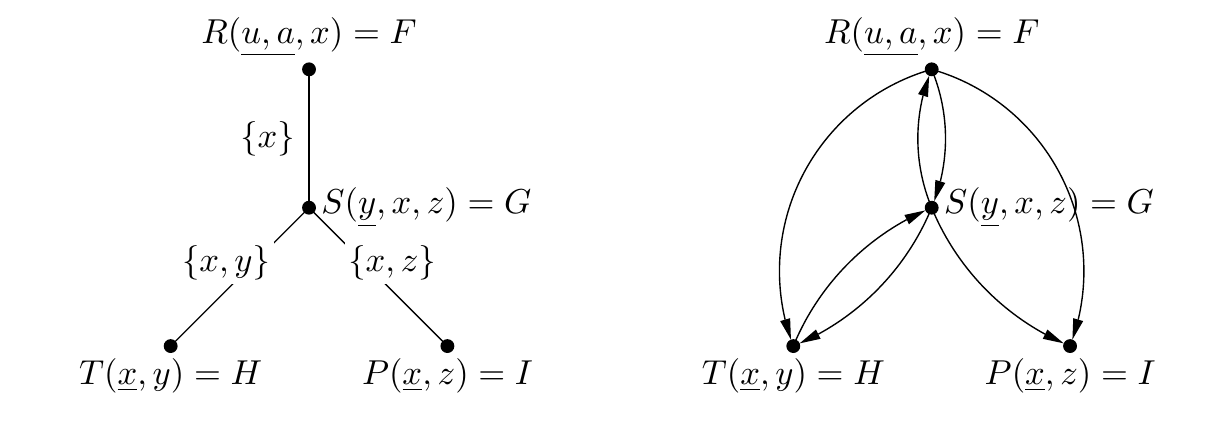}
\caption{
Join tree (left) and attack graph (right) of query $q_{\ax}$.
The attack from $G$ to $F$ is strong. All other attacks are weak.}\label{fig:ax}
\end{figure*}

The primary key of an atom $F$ gives rise to a functional dependency among the variables that occur in $F$.
For example, $R(\underline{x,y},z,u)$ gives rise to $\fd{\{x,y\}}{\{x,y,z,u\}}$, which will be abbreviated as $\fd{xy}{xyzu}$ (and which is equivalent to $\fd{xy}{zu}$).
The set $\FD{q}$ defined next collects all functional dependencies that arise in atoms of $q$.

\begin{definition}
Let $q$ be a Boolean conjunctive query.
We define $\FD{q}$ as the following set of functional dependencies. 

\myhfill$\FD{q}=\{\fd{\key{F}}{\all{F}}\mid F\in q\}$
\end{definition}


Concerning the following definition, recall from relational database theory~\cite[page 387]{DBLP:books/cs/Ullman88} that if $\Sigma$ is a set of functional dependencies over a set $U$ of attributes and $X\subseteq U$,  
then the attribute closure of $X$ (with respect to $\Sigma$) is the set $\{A\in U\mid\Sigma\models\fd{X}{A}\}$.

\begin{definition}
Let $q$ be a Boolean conjunctive query.
For every $F\in q$, we define $\keycl{F}{q}$ as the following set of variables. 

\myhfill$\keycl{F}{q}=\{x\in\queryvars{q}\mid\FD{q\setminus\{F\}}\models\fd{\key{F}}{x}\}$
\end{definition}

In words, $\keycl{F}{q}$ is the attribute closure of the set $\key{F}$ with respect to the set of functional dependencies that arise in the atoms of $q\setminus\{F\}$.
Note that variables play the role of attributes in our framework.

\begin{example}\label{ex:keycl}
Let $q_{\ax}=\{R(\underline{u,a},x)$, $S(\underline{y},x,z)$, $T(\underline{x},y)$, $P(\underline{x},z)\}$.
A join tree for this query is shown in Fig.~\ref{fig:ax} (left).
To shorten notation, let $F=R(\underline{u,a},x)$, $G=S(\underline{y},x,z)$, $H=T(\underline{x},y)$, and $I=P(\underline{x},z)$, as indicated in the figure.
We have the following.

\myhfill$
\begin{array}[b]{l@{\mbox{\ and\ }}l}
\multicolumn{2}{l}{
\FD{q_{\ax}\setminus\{F\}}=\{\fd{y}{xyz}, \fd{x}{xy}, \fd{x}{xz}\}}\\[0.6ex]
\key{F}=\{u\} & \keycl{F}{q_{\ax}}=\{u\}\\[1.5ex]
\multicolumn{2}{l}{
\FD{q_{\ax}\setminus\{G\}}=\{\fd{u}{ux}, \fd{x}{xy}, \fd{x}{xz}\}}\\[0.6ex]
\key{G}=\{y\} & \keycl{G}{q_{\ax}}=\{y\}\\[1.5ex]
\multicolumn{2}{l}{
\FD{q_{\ax}\setminus\{H\}}=\{\fd{u}{ux}, \fd{y}{xyz}, \fd{x}{xz}\}}\\[0.6ex]
\key{H}=\{x\} & \keycl{H}{q_{\ax}}=\{x,z\}\\[1.5ex]
\multicolumn{2}{l}{
\FD{q_{\ax}\setminus\{I\}}=\{\fd{u}{ux}, \fd{y}{xyz}, \fd{x}{xy}\}}\\[0.6ex]
\key{I}=\{x\} & \keycl{I}{q_{\ax}}=\{x,y,z\}
\end{array}
$
\end{example}

\begin{definition}
Let $q$ be an acyclic Boolean conjunctive query.
Let $\tau$ be a join tree for $q$.
The {\em attack graph\/} of $\tau$ is a directed graph whose vertices are the atoms of $q$.
There is a directed edge from $F$ to $G$ if $F,G$ are distinct atoms such that for every label $L$ on the unique path that links $F$ and $G$ in $\tau$,
we have $L\nsubseteq\keycl{F}{q}$.

We write $\attacks{F}{G}{\tau}$ if the attack graph of $\tau$ contains a directed edge from $F$ to $G$.
The directed edge $\attacks{F}{G}{\tau}$ is also called an {\em attack from $F$ to $G$\/}.
If $\attacks{F}{G}{\tau}$, we say that $F$ {\em attacks\/} $G$ (or that $G$ is attacked by $F$).
\end{definition}

\begin{example}
This is a continuation of Example~\ref{ex:keycl}.
Fig.~\ref{fig:ax} (left) shows a join tree $\tau_{\ax}$ for query $q_{\ax}$.
The attack graph of $\tau_{\ax}$ is shown in Fig.~\ref{fig:ax} (right) and is computed as follows. 

Let us first compute the attacks outgoing from $F$.
The path from $F$ to $G$ in the join tree is $F\step{\{x\}}G$. 
Since the label $\{x\}$ is not contained in $\keycl{F}{q_{\ax}}$, the attack graph contains a directed edge from $F$ to $G$, i.e., $\attacks{F}{G}{\tau_{\ax}}$.
The path from $F$ to $H$ in the join tree is $F\step{\{x\}}G\step{\{x,y\}}H$.
Since no label on that path is contained in $\keycl{F}{q_{\ax}}$, the attack graph contains a directed edge from $F$ to $H$.
In the same way, one finds that $F$ attacks $I$.

Let us next compute the attacks outgoing from $H$.
The path from $H$ to $G$ in the join tree is $H\step{\{x,y\}}G$. 
Since the label $\{x,y\}$ is not contained in $\keycl{G}{q_{\ax}}$, the attack graph contains a directed edge from $H$ to $G$, .i.e., $\attacks{H}{G}{\tau_{\ax}}$.
The path from $H$ to $F$ in the join tree is $H\step{\{x,y\}}G\step{\{x\}}F$.
Since the label $\{x\}$ is contained in $\keycl{H}{q_{\ax}}$,
the attack graph contains no directed edge from $H$ to $F$.
And so on.
The complete attack graph is shown in Fig.~\ref{fig:ax} (right).
\end{example}

Remarkably, it was shown in~\cite{DBLP:journals/tods/Wijsen12} that if $\tau_{1}$ and $\tau_{2}$ are distinct join trees for the same acyclic query $q$,
then the attack graph of $\tau_{1}$ is identical to the attack graph of $\tau_{2}$.
This motivates the following definition.

\begin{definition}
Let $q$ be an acyclic Boolean conjunctive query.
The attack graph of $q$ is the attack graph of $\tau$ for any join tree $\tau$ for $q$.
We write $\attacks{F}{G}{q}$ (or simply $\attacksq{F}{G}{q}$ if $q$ is clear from the context) to indicate that the attack graph of $q$ contains a directed edge from $F$ to $G$.
We write  $\nattacks{F}{G}{q}$ if it is not the case that $\attacks{F}{G}{q}$.
\end{definition}

The attack graph of an acyclic query $q$ can be computed in quadratic time in the length of $q$~\cite{DBLP:journals/tods/Wijsen12}.
Figures~\ref{fig:threecycles} and~\ref{fig:acqthree} show attack graphs, but omit join trees.
The main result in~\cite{DBLP:journals/tods/Wijsen12} is the following.

\begin{theorem}[\cite{DBLP:journals/tods/Wijsen12}]\label{the:acyclic}
The following are equivalent for all acyclic Boolean conjunctive queries $q$ without self-join:
\begin{enumerate}
\item
The attack graph of $q$ is acyclic.
\item
$\cqak{q}$ is first-order expressible.
\end{enumerate}
\end{theorem}

Finally, we provide two lemmas that will be useful later on.

\begin{lemma}\label{lem:rsu}
Let $q$ be an acyclic Boolean conjunctive query.
Let $F,G$ be distinct atoms of $q$.
If $\attacksq{F}{G}{q}$, then $\key{G}\nsubseteq\keycl{F}{q}$ and $\all{F}\nsubseteq\keycl{F}{q}$.
\end{lemma}

\begin{lemma}[\cite{DBLP:journals/tods/Wijsen12}]\label{lem:A2}
Let $q$ be an acyclic Boolean conjunctive query.
Let $F,G,H$ be distinct atoms of $q$.
If $\attacksq{F}{G}{q}$ and $\attacksq{G}{H}{q}$,
then $\attacksq{F}{H}{q}$ or $\attacksq{G}{F}{q}$.
\end{lemma}

\section{Intractability}\label{sec:intractability}

The following definition classifies cycles in attack graphs as either strong or weak.
The main result of this section is that $\cqak{q}$ is \coNP-complete for acyclic queries $q$ whose attack graph contains a strong cycle.

\begin{definition}
Let $q$ be an acyclic Boolean conjunctive query.
For every $F\in q$, we define $\keyclsup{F}{q}$ as the following set of variables. 
$$\keyclsup{F}{q}=\{x\in\queryvars{q}\mid\FD{q}\models\fd{\key{F}}{x}\}$$
An attack $\attacksq{F}{G}{q}$ in the attack graph of $q$ is called {\em weak\/} if $\key{G}\subseteq\keyclsup{F}{q}$.
An attack that is not weak, is called {\em strong\/}.

A (directed) {\em cycle of size $n$\/} in the attack graph of $q$ is a sequence of edges
$F_{0}\attacksymbolq{q}F_{1}\attacksymbolq{q}F_{2}\dots\attacksymbolq{q}F_{n-1}\attacksymbolq{q}F_{0}$
such that $i\neq j$ implies $F_{i}\neq F_{j}$.
Thus, {\em cycle\/} means elementary cycle.

A cycle in the attack graph of $q$ is called {\em strong\/} if at least one attack in the cycle is strong. 
A cycle that is not strong, is called {\em weak\/}.
\end{definition}

It is straightforward that $\keycl{F}{q}\subseteq\keyclsup{F}{q}$.

\begin{example}\label{ex:cycles}
For the query $q_{\ax}$ in  Fig.~\ref{fig:ax}, we have the following.
\begin{eqnarray*}
\FD{q_{\ax}} & = & \{\fd{u}{ux}, \fd{y}{xyz}, \fd{x}{xy}, \fd{x}{xz}\}\\
\keyclsup{F}{q_{\ax}} & = & \{u,x,y,z\}\\
\keyclsup{G}{q_{\ax}} & = & \{x,y,z\}\\
\keyclsup{H}{q_{\ax}} & = & \{x,y,z\}\\
\keyclsup{I}{q_{\ax}} & = & \{x,y,z\}
\end{eqnarray*}
The attack $\attacks{F}{G}{q_{\ax}}$ is weak, because $\key{G}=\{x\}\subseteq\keyclsup{F}{q_{\ax}}$.
The attack $\attacks{G}{F}{q_{\ax}}$ is strong, because $\key{F}=\{u\}\not\subseteq\keyclsup{G}{q_{\ax}}$.
One can verify that the attack from $G$ to $F$ is the only strong attack in the attack graph of $q_{\ax}$.

The attack cycle $G\attacksymbol{q_{\ax}}H\attacksymbol{q_{\ax}}G$ is weak.
The attack cycle $F\attacksymbol{q_{\ax}}G\attacksymbol{q_{\ax}}F$ is strong, because it contains the strong attack $\attacks{G}{F}{q_{\ax}}$.
For the same reason, the attack cycle $F\attacksymbol{q_{\ax}}H\attacksymbol{q_{\ax}}G\attacksymbol{q_{\ax}}F$ is strong.
\end{example}

Example~\ref{ex:cycles} showed that the attack graph of $q_{\ax}$ has a strong cycle of length~$3$, and a strong cycle of length~$2$.
This is no coincidence, as stated by the following lemma.

\begin{lemma}\label{lem:two}
Let $q$ be an acyclic Boolean conjunctive query.
If the attack graph of $q$ contains a strong cycle,
then it contains a strong cycle of length $2$.
\end{lemma}

\begin{figure}\centering
\includegraphics{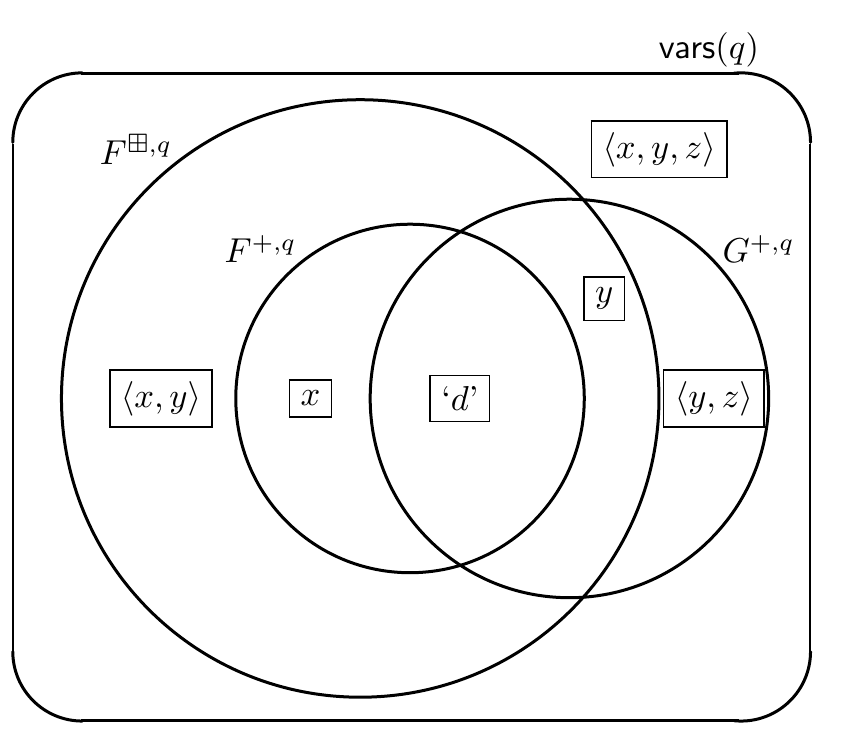}
\caption{Help for the proof of Theorem~\ref{the:strongcycle}.}\label{fig:schema}
\end{figure}

The following proof establishes that for every acyclic query $q$ whose attack graph contains a strong cycle,
there exists a polynomial-time many-one reduction
from $\cqak{q_{0}}$ to $\cqak{q}$, where $q_{0}=\{R_{0}(\underline{x},y), S_{0}(\underline{y,z},x)\}$.
Since $\cqak{q_{0}}$ was proved \coNP-hard by Kolaitis and Pema~\cite{DBLP:journals/ipl/KolaitisP12},
we obtain the desired \coNP-hard lower bound for $\cqak{q}$.
As the proof is rather involved, we provide in Fig.~\ref{fig:schema} a mnemonic for the construction in the beginning of the proof.
To further improve readability, some parts of the proof will be stated as sublemmas.

\begin{theorem}\label{the:strongcycle}
Let $q$ be an acyclic Boolean conjunctive query without self-join.
If the attack graph of $q$ contains a strong cycle,
then $\cqak{q}$ is {\bf coNP}-complete.
\end{theorem}
\begin{proof}
Since $\cqak{q}$ is obviously in \coNP, it suffices to show that it is \coNP-hard.
Assume that the attack graph of $q$ contains a strong cycle.
By Lemma~\ref{lem:two}, we can assume $F,G\in q$ such that $F\attacksymbolq{q}G\attacksymbolq{q}F$ and the attack $\attacksq{F}{G}{q}$ is strong.
For every valuation $\theta$ over $\{x,y,z\}$,
we define $\rv{\theta}$ as the following valuation over $\queryvars{q}$.
\begin{enumerate}
\item
If $u\in\keycl{F}{q}\cap\keycl{G}{q}$, then $\rv{\theta}(u)=\con{$d$}$ for some fixed constant $d$;
\item
if $u\in\keycl{F}{q}\setminus\keycl{G}{q}$, then $\rv{\theta}(u)=\theta(x)$;
\item
if $u\in\keycl{G}{q}\setminus\keyclsup{F}{q}$, then $\rv{\theta}(u)=\skolem{\theta(y)}{\theta(z)}$;
\item
if $u\in\bigformula{\keycl{G}{q}\cap\keyclsup{F}{q}}\setminus\keycl{F}{q}$, then $\rv{\theta}(u)=\theta(y)$;
\item
if $u\in\keyclsup{F}{q}\setminus\bigformula{\keycl{F}{q}\cup\keycl{G}{q}}$, then $\rv{\theta}(u)=\skolem{\theta(x)}{\theta(y)}$; and
\item
if $u\not\in\keyclsup{F}{q}\cup\keycl{G}{q}$, then $\rv{\theta}(u)=\triskolem{\theta(x)}{\theta(y)}{\theta(z)}$.
\end{enumerate}
Notice that $\rv{\theta}(u)$ can be a sequence of length two or three;
two sequences of the same length are equal if they contain the same elements in the same order.
The Venn diagram of Fig.~\ref{fig:schema} will come in handy:
every region contains a boxed label that indicates how $\rv{\theta}(u)$ is computed for variables $u$ in that region.
For example, assume $u$ belongs to the region with label $\skolem{x}{y}$ 
(i.e., $u\in\keyclsup{F}{q}\setminus\bigformula{\keycl{F}{q}\cup\keycl{G}{q}}$),
then $\rv{\theta}(u)=\skolem{\theta(x)}{\theta(y)}$.

We show three sublemmas that will be used later on in the proof.

\begin{sublemma}\label{prop:H}
Let $\theta_{1},\theta_{2}$ be two valuations over $\{x,y,z\}$.
If $H\in q$ such that $F\neq H\neq G$,
then $\{\rv{\theta_{1}}(H),\rv{\theta_{2}}(H)\}$ is consistent.
\end{sublemma}
\begin{subproof}[Sublemma~\ref{prop:H}]
Let $H\in q$ such that $F\neq H\neq G$.
Assume the following.
\begin{equation}\label{eq:assumekeyeq}
\mbox{For every $u\in\key{H}$}, \rv{\theta_{1}}(u)=\rv{\theta_{2}}(u).
\end{equation}
It suffices to show the following.
\begin{equation}\label{eq:showeq}
\mbox{For every $u\in\all{H}$}, \rv{\theta_{1}}(u)=\rv{\theta_{2}}(u).
\end{equation}
We consider four cases.

\myparagraph{Case $\theta_{1}(x)=\theta_{2}(x)$ and $\theta_{1}(y)=\theta_{2}(y)$.}
If $\theta_{1}(z)=\theta_{2}(z)$, then $\theta_{1}=\theta_{2}$, 
and~(\ref{eq:showeq}) holds vacuously.
Assume next $\theta_{1}(z)\neq\theta_{2}(z)$.
Then it follows from~(\ref{eq:assumekeyeq}) that no variable of $\key{H}$ belongs to a region of the Venn diagram (see Fig.~\ref{fig:schema}) that contains $z$.
Since $z$ occurs in all regions outside $\keyclsup{F}{q}$, 
we conclude $\key{H}\subseteq\keyclsup{F}{q}$.
Since $\FD{q}$ contains $\fd{\key{H}}{\all{H}}$, it follows $\all{H}\subseteq\keyclsup{F}{q}$.
Since $z$ does not occur inside $\keyclsup{F}{q}$ in the Venn diagram, we conclude~(\ref{eq:showeq}).

\myparagraph{Case $\theta_{1}(x)=\theta_{2}(x)$ and $\theta_{1}(y)\neq\theta_{2}(y)$.}
By~(\ref{eq:assumekeyeq}), no variable of $\key{H}$ belongs to a region of the Venn diagram that contains $y$.
It follows $\key{H}\subseteq\keycl{F}{q}$.
Consequently,  $\all{H}\subseteq\keycl{F}{q}$.
Since neither $y$ nor $z$ occurs inside $\keycl{F}{q}$ in the Venn diagram, we conclude~(\ref{eq:showeq}).

\myparagraph{Case $\theta_{1}(x)\neq\theta_{2}(x)$ and $\theta_{1}(y)=\theta_{2}(y)$.}
First assume $\theta_{1}(z)=\theta_{2}(z)$.
By~(\ref{eq:assumekeyeq}), no variable of $\key{H}$ belongs to a region of the Venn diagram that contains $x$.
Consequently, $\key{H}\subseteq\keycl{G}{q}$.
It follows $\all{H}\subseteq\keycl{G}{q}$.
Since $x$ does not occur inside $\keycl{G}{q}$ in the Venn diagram,
we conclude~(\ref{eq:showeq}).

Next assume $\theta_{1}(z)\neq\theta_{2}(z)$.
By~(\ref{eq:assumekeyeq}), no variable of $\key{H}$ belongs to a region of the Venn diagram that contains $x$ or $z$.
Consequently, $\key{H}\subseteq\keyclsup{F}{q}\cap\keycl{G}{q}$.
It follows $\all{H}\subseteq\keyclsup{F}{q}\cap\keycl{G}{q}$.
Since neither $x$ nor $z$ occurs inside $\keyclsup{F}{q}\cap\keycl{G}{q}$ in the Venn diagram, we conclude~(\ref{eq:showeq}).

\myparagraph{Case $\theta_{1}(x)\neq\theta_{2}(x)$ and $\theta_{1}(y)\neq\theta_{2}(y)$.}
By~(\ref{eq:assumekeyeq}), no variable of $\key{H}$ belongs to a region of the Venn diagram that contains $x$ or $y$.
Consequently, $\key{H}\subseteq\keycl{F}{q}\cap\keycl{G}{q}$.
It follows $\all{H}\subseteq\keycl{F}{q}\cap\keycl{G}{q}$.
Since none of $x$, $y$, or $z$ occurs inside $\keycl{F}{q}\cap\keycl{G}{q}$ in the Venn diagram,
we conclude~(\ref{eq:showeq}).
This concludes the proof of Sublemma~\ref{prop:H}.
\end{subproof}

\begin{sublemma}\label{prop:F}
Let $\theta_{1},\theta_{2}$ be two valuations over $\{x,y,z\}$.
\begin{enumerate}
\item
$\rv{\theta_{1}}(F)$ and $\rv{\theta_{2}}(F)$ are key-equal
$\iff$
$\theta_{1}(x)=\theta_{2}(x)$.
\item
$\rv{\theta_{1}}(F)=\rv{\theta_{2}}(F)$
$\iff$
$\theta_{1}(x)=\theta_{2}(x)$ and $\theta_{1}(y)=\theta_{2}(y)$.
\end{enumerate}
\end{sublemma}

\begin{sublemma}\label{prop:G}
Let $\theta_{1},\theta_{2}$ be two valuations over $\{x,y,z\}$.
\begin{enumerate}
\item
$\rv{\theta_{1}}(G)$ and $\rv{\theta_{2}}(G)$ are key-equal
$\iff$
$\theta_{1}(y)=\theta_{2}(y)$ and $\theta_{1}(z)=\theta_{2}(z)$.
\item
$\rv{\theta_{1}}(G)=\rv{\theta_{2}}(G)$
$\iff$
$\theta_{1}=\theta_{2}$.
\end{enumerate}
\end{sublemma}

We continue the proof of Theorem~\ref{the:strongcycle}.
Let $q_{0}=\{R_{0}(\underline{x},y)$, $S_{0}(\underline{y,z},x)\}$.
The signatures of $R_{0}$ and $S_{0}$ are $\signature{2}{1}$ and $\signature{3}{2}$ respectively.
Let $F_{0}=R_{0}(\underline{x},y)$ and $G_{0}=S_{0}(\underline{y,z},x)$.
In the remainder of the proof,
we establish a polynomial-time many-one reduction from $\cqak{q_{0}}$ to $\cqak{q}$.
\coNP-hardness of $\cqak{q}$ then follows from \coNP-hardness of $\cqak{q_{0}}$, which was established in~\cite{DBLP:journals/ipl/KolaitisP12}.

Let $\db_{0}$ be an uncertain database.
By Lemma~\ref{lem:purified}, we can assume that $\db_{0}$ is purified relative to $q_{0}$.
Let $\calV$ be the set of valuations $\theta$ over $\{x,y,z\}$ such that $\theta(q_{0})\subseteq\db_{0}$. 
Since $\db_{0}$ is purified, the following holds.
\begin{eqnarray*}
\db_{0} & = & \{\theta(F_{0})\mid\theta\in\calV\}\cup\{\theta(G_{0})\mid\theta\in\calV\}
\end{eqnarray*}        
Let $\db=\{\rv{\theta}(H)\mid H\in q, \theta\in\calV\}$.
Since $\calV$ can be computed in polynomial time in the size of $\db_{0}$, the reduction from $\db_{0}$ to $\db$ is in polynomial time.
Since $q$ contains no self-join, the set $\db$ is partitioned by the three disjoint subsets defined next.
\begin{eqnarray*}
\db_{F} & = & \{\rv{\theta}(F)\mid\theta\in\calV\}\\
\db_{G} & = & \{\rv{\theta}(G)\mid\theta\in\calV\}\\ 
\db_{\rest} & = & \{\rv{\theta}(H)\mid H\in q, F\neq H\neq G, \theta\in\calV\}
\end{eqnarray*}
Since $\db_{\rest}$ is consistent by Sublemma~\ref{prop:H}, every repair of $\db$ is the disjoint union of $\db_{\rest}$, a repair of $\db_{F}$, and a repair of $\db_{G}$.
In the next step of the proof, we establish a one-to-one relationship between repairs of $\db_{0}$ and repairs of $\db$.

The function $\transformsymbol$ will map repairs of $\db_{0}$ to repairs of $\db$. 
For every repair $\rep_{0}$ of $\db_{0}$, $\transform{\rep_{0}}$ is the disjoint union of three sets, as follows. 
$$
\begin{array}{rccl}
\transform{\rep_{0}} 
                     & = &      & \{\rv{\theta}(F)\mid\theta(F_{0})\in\rep_{0},\theta\in\calV\}\\
                     &   & \cup & \{\rv{\theta}(G)\mid\theta(G_{0})\in\rep_{0},\theta\in\calV\}\\
                     &   & \cup & \db_{\rest}
\end{array}
$$      
Clearly, the first of these three sets is contained in $\db_{F}$, and the second in $\db_{G}$.               
By Sublemmas~\ref{prop:F} and~\ref{prop:G}, for every $\theta\in\calV$,
\begin{eqnarray}
\theta(F_{0})\in\rep_{0} & \iff & \rv{\theta}(F)\in\transform{\rep_{0}}\label{eq:F}\\
\theta(G_{0})\in\rep_{0} & \iff & \rv{\theta}(G)\in\transform{\rep_{0}}\label{eq:G}
\end{eqnarray}
To prove the $\impliedby$-direction of (\ref{eq:F}) (the other implications are straightforward), 
assume $A\in\transform{\rep_{0}}$ with $A=\rv{\theta}(F)$.
By the definition of $\transformsymbol$, we can assume $\theta'\in\calV$ such that $\theta'(F_{0})\in\rep_{0}$ and $\rv{\theta'}(F)=A$.
From $\rv{\theta}(F)=\rv{\theta'}(F)$, it follows by Sublemma~\ref{prop:F} that $\theta(F_{0})=\theta'(F_{0})$, hence $\theta(F_{0})\in\rep_{0}$.

The following sublemma states that $\transformsymbol$ is a bijection from the set of repairs of $\db_{0}$ to the set of repairs of $\db$. 
\begin{sublemma}\label{prop:bijection}
\begin{enumerate}
\item
If $\rep_{0}$ is a repair of $\db_{0}$,
then $\transform{\rep_{0}}$ is a repair of $\db$.
\item
For every repair $\rep$ of $\db$,
there exists a repair $\rep_{0}$ of $\db_{0}$ such that $\rep=\transform{\rep_{0}}$.
\item
If $\rep_{0},\rep_{0}'$ are distinct repairs of $\db_{0}$,
then $\transform{\rep_{0}}\neq\transform{\rep_{0}'}$.
\end{enumerate}
\end{sublemma}

To conclude the proof of Theorem~\ref{the:strongcycle}, we show:
$$
\db_{0}\in\cqak{q_{0}}\iff\db\in\cqak{q}.
$$
By Sublemma~\ref{prop:bijection}, it is sufficient to prove that for every repair $\rep_{0}$ of $\db_{0}$,
$$
\rep_{0}\models q_{0}
\iff
\transform{\rep_{0}}\models q.
$$
\framebox{$\implies$}
Assume $\rep_{0}\models q_{0}$.
We can assume $\theta\in\calV$ such that $\theta(q_{0})\subseteq\rep_{0}$.
Obviously, $\rv{\theta}(q)\subseteq\transform{\rep_{0}}$.

\framebox{$\impliedby$}
Assume $\transform{\rep_{0}}\models q$.
We can assume a valuation $\mu$ over $\queryvars{q}$ such that $\mu(q)\subseteq\transform{\rep_{0}}$.

Let $\tau$ be a join tree for $q$.
Let $H_{0}\step{L_{1}}H_{1}\dots\step{L_{\ell}}H_{\ell}$ be the unique path in $\tau$ between $F$ and $G$,
where $H_{0}=F$ and $H_{\ell}=G$.
For $i\in\{0,\dots,\ell\}$, we can assume $\theta_{i}\in\calV$ such that $\mu(H_{i})=\rv{\theta_{i}}(H_{i})\in\transform{\rep_{0}}$.
Let $i\in\{0,\dots,\ell-1\}$.
We show $\theta_{i}(x)=\theta_{i+1}(x)$ and $\theta_{i}(y)=\theta_{i+1}(y)$.
Since $F\attacksymbolq{q}G\attacksymbolq{q}F$,
the label $L_{i}$ contains a variable $u_{i}$ such that $u_{i}\not\in\keycl{F}{q}$ and a variable $w_{i}$ such that $w_{i}\not\in\keycl{G}{q}$ (possibly $u_{i}=w_{i}$).

Since $u_{i}\in\all{H_{i}}\cap\all{H_{i+1}}$, it must be the case that $\rv{\theta_{i}}(u_{i})=\mu(u_{i})=\rv{\theta_{i+1}}(u_{i})$.
Since $y$ occurs in every region outside $\keycl{F}{q}$ in the Venn diagram (Fig.~\ref{fig:schema}) and $u_{i}\not\in\keycl{F}{q}$,
it is correct to conclude $\theta_{i}(y)=\theta_{i+1}(y)$.

Likewise, since $w_{i}\in\all{H_{i}}\cap\all{H_{i+1}}$, it must be the case that $\rv{\theta_{i}}(w_{i})=\mu(w_{i})=\rv{\theta_{i+1}}(w_{i})$.
Since $x$ occurs in every region outside $\keycl{G}{q}$ in the Venn diagram and $w_{i}\not\in\keycl{G}{q}$,
it is correct to conclude $\theta_{i}(x)=\theta_{i+1}(x)$.

Consequently, $\theta_{0}(x)=\theta_{\ell}(x)$ and $\theta_{0}(y)=\theta_{\ell}(y)$.
From $\rv{\theta_{0}}(H_{0}),\rv{\theta_{\ell}}(H_{\ell})\in\transform{\rep_{0}}$, $H_{0}=F$, and $H_{\ell}=G$,
it follows $\theta_{0}(F_{0}),\theta_{\ell}(G_{0})\in\rep_{0}$ by~(\ref{eq:F}) and~(\ref{eq:G}). 
Since $\theta_{0}$ and $\theta_{\ell}$ agree on each variable in $\all{F_{0}}\cap\all{G_{0}}=\{x,y\}$,
it follows $\rep_{0}\models q_{0}$.
This concludes the proof of Theorem~\ref{the:strongcycle}.
\end{proof}

\section{Tractability}\label{sec:tractability}

We conjecture that if the attack graph of an acyclic query $q$ contains no strong cycle,
then $\cqak{q}$ is in \PTIME.

\begin{conjecture}\label{con:weak}
Let $q$ be an acyclic Boolean conjunctive query without self-join.
If all cycles in the attack graph of $q$ are weak,
then $\cqak{q}$ is in \PTIME.
\end{conjecture}

Notice that by Theorem~\ref{the:acyclic}, we know that Conjecture~\ref{con:weak} holds in the special case where $q$'s attack graph contains no cycle at all. 
Theorem~\ref{the:strongcycle} and Conjecture~\ref{con:weak} together imply that for every acyclic query $q$, $\cqak{q}$ is either in \PTIME\ or \coNP-complete.
In the following section, a somewhat weaker version of Conjecture~\ref{con:weak} is proved.

\subsection{All Cycles are Weak and Terminal}

\begin{figure}\centering
\includegraphics{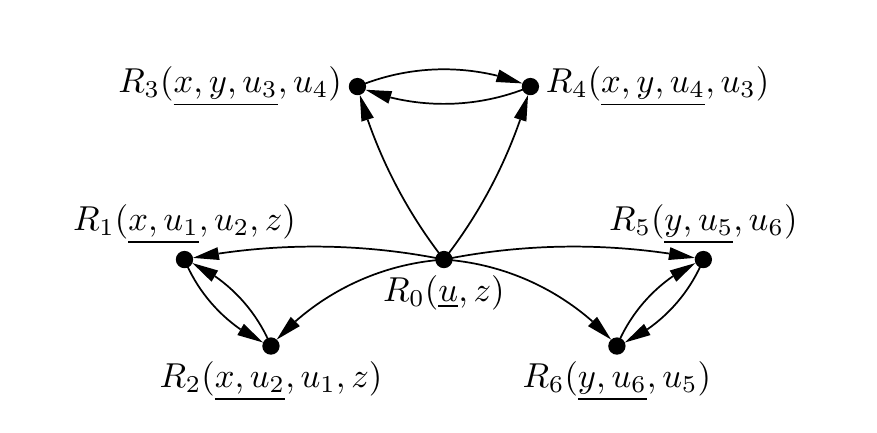}
\caption{Attack graph. All cycles are weak and terminal.}\label{fig:threecycles}
\end{figure}

We show a weaker version of Conjecture~\ref{con:weak}.
In this weaker version, the premise ``all cycles are weak" is strengthened into ``all cycles are weak and terminal."

\begin{definition}
A cycle in a directed graph is called {\em terminal\/} if the graph contains no directed edge from a vertex in the cycle to a vertex outside the cycle.
A cycle is {\em nonterminal\/} if it is not terminal.
\end{definition}

\begin{example}
Figure~\ref{fig:threecycles} shows the attack graph of the acyclic query
$\{R_{1}(\underline{x,u_{1}},u_{2},z)$, 
$R_{2}(\underline{x,u_{2}},u_{1},z)$,
$R_{3}(\underline{x,y,u_{3}},u_{4})$, 
$R_{4}(\underline{x,y,u_{4}},u_{3})$,
$R_{5}(\underline{y,u_{5}},u_{6})$, 
$R_{6}(\underline{y,u_{6}},u_{5})\}$.
All attack cycles are terminal and weak.
\end{example}

\begin{example}
In the attack graph of Fig.~\ref{fig:acqthree}, all cycles are weak, but no cycle is terminal.
\end{example}

\begin{theorem}\label{the:outdegree}
Let $q$ be an acyclic Boolean conjunctive query without self-join.
If all cycles in the attack graph of $q$ are weak and terminal,
then $\cqak{q}$ is in \PTIME.
\end{theorem}

Notice that if a query $q$ has exactly two atoms, then $q$ is acyclic and every cycle in $q$'s attack graph must be terminal.
Therefore Theorems~\ref{the:strongcycle} and~\ref{the:outdegree} together imply the dichotomy theorem of Kolaitis and Pema~\cite{DBLP:journals/ipl/KolaitisP12}.

To prove Theorem~\ref{the:outdegree}, we need four helping lemmas.
In simple words, the first lemma states that if we replace a variable with a constant in an acyclic query,
then no new attacks are generated, and weak attacks cannot become strong.

\begin{definition}
Let $q$ be a Boolean conjunctive query.
If $\vec{x}=\tuple{x_{1},\dots,x_{\ell}}$ is a sequence of distinct variables
and $\vec{a}=\tuple{a_{1},\dots,a_{\ell}}$ a sequence of constants,
then $\substitute{q}{\vec{x}}{\vec{a}}$ denotes the query obtained from $q$ by replacing each occurrence of $x_{i}$ with $a_{i}$, for all $i\in\{1,\dots,\ell\}$.
If $\theta$ is a valuation, then $\substitute{\theta}{\vec{x}}{\vec{a}}$ is the valuation such that $\substitute{\theta}{\vec{x}}{\vec{a}}(\vec{x})=\vec{a}$ and $\substitute{\theta}{\vec{x}}{\vec{a}}(y)=\theta(y)$ if $y\not\in\sequencevars{\vec{x}}$.
\end{definition}

\begin{lemma}\label{lem:weakweak}
Let $q$ be an acyclic Boolean conjunctive query without self-join.
Let $F,G\in q$.
Let $z\in\queryvars{q}$ and let $c$ be a constant.
Let $q'=\substitute{q}{z}{c}$, $F'=\substitute{F}{z}{c}$, and $G'=\substitute{G}{z}{c}$.
Then, the following hold.
\begin{enumerate}
\item
$q'$ is acyclic.
\item
If $\attacks{F'}{G'}{q'}$, then $\attacks{F}{G}{q}$.
\item
If  $\attacks{F'}{G'}{q'}$ and $\attacks{F}{G}{q}$ is a weak attack, then $\attacks{F'}{G'}{q'}$ is a weak attack.
\end{enumerate}
\end{lemma}

\begin{lemma}\label{lem:terminalcycles}
Let $q$ be an acyclic Boolean conjunctive query.
If each cycle in the attack graph of $q$ is terminal, then each cycle in the attack graph has length $2$. 
\end{lemma}

\begin{lemma}\label{lem:sharedvariables}
Let $q$ be an acyclic Boolean conjunctive query such that each cycle of the attack graph of $q$ is terminal and each atom of $q$ belongs to a cycle of the attack graph.
\begin{enumerate}
\item
If the same variable $x$ occurs in two distinct cycles of the attack graph,
then for each atom $F$ in these cycles, $x\in\key{F}$.
\item
If $\attacks{F}{G}{q}$ is a weak attack, then $\key{G}\subseteq\all{F}$. 
\end{enumerate}
\end{lemma}

The following lemma applies to queries with an atom whose primary key contains no variables.

\begin{lemma}\label{lem:SE3}
Let $q$ be a Boolean conjunctive query without self-join.
Let $F\in q$ such that $\key{F}=\emptyset$.
Let $q'=q\setminus\{F\}$.
Let $\vec{y}$ be a sequence of distinct variables such that $\sequencevars{\vec{y}}=\all{F}$.
Let $\db$ be an uncertain database that is purified relative to $q$, and let $D$ be the active domain of $\db$. 
Then the following are equivalent:
\begin{enumerate}
\item
$\db\in\cqak{q}$.
\item 
$\db\neq\emptyset$ and for all $\vec{b}\in D^{\card{\vec{y}}}$,
if $\substitute{F}{\vec{y}}{\vec{b}}\in\db$,
then $\db\in\cqak{\substitute{q'}{\vec{y}}{\vec{b}}}$.
\end{enumerate}
\end{lemma}

The proof of Theorem~\ref{the:outdegree} can now be given.

\myparagraph{}
\begin{proof}[Theorem~\ref{the:outdegree}]
Given uncertain database $\db$, we need to show that it can be decided in polynomial time (in the size of $\db$) whether $\db\in\cqak{q}$.
Let $D$ be the active domain of $\db$.
By Lemma~\ref{lem:purified}, we can assume that $\db$ is purified relative to $q$.

The proof runs by induction on the length of $q$.
For the base of the induction,
we consider the case where the attack graph of $q$ contains no unattacked atom (i.e., no atom  has zero indegree).
$\cqak{q}$ is obviously in \PTIME\ if $q=\{\}$.
Assume next that $q$ is nonempty.

Since all cycles of $q$'s attack graph are terminal and every atom has an incoming attack,
every atom of $q$ belongs to some cycle of the attack graph.
By Lemma~\ref{lem:terminalcycles}, the attack graph of $q$ is a set of disjoint weak cycles
$F_{1}\attacksymbolq{q}G_{1}\attacksymbolq{q}F_{1}$, \dots, $F_{\ell}\attacksymbolq{q}G_{\ell}\attacksymbolq{q}F_{\ell}$ for some $\ell\geq 1$.
For $i\in\{1,\dots,\ell\}$, let $q_{i}=\{F_{i},G_{i}\}$, and let $\vec{x}_{i}$ be a sequence of distinct variables that contains every variable $x\in\queryvars{q_{i}}$ such that for some $j\neq i$, $x\in\queryvars{q_{j}}$.
By Lemma~\ref{lem:sharedvariables}, $\sequencevars{x_{i}}\subseteq\key{F_{i}}\cap\key{G_{i}}$.

For $i\in\{1,\dots,\ell\}$, let $\db_{i}$ be the subset of $\db$ containing every fact $A$ with the same relation name as $F_{i}$ or $G_{i}$.
Call a {\em partition\/} of $\db_{i}$ a maximal subset $P$ of $\db_{i}$ such that for some $\vec{a}\in D^{\card  {\vec{x}_{i}}}$, for all $A\in P$, there exists a valuation $\theta$ such that $A=\substitute{\theta}{\vec{x}_{i}}{\vec{a}}(F_{i})$ or $A=\substitute{\theta}{\vec{x}_{i}}{\vec{a}}(G_{i})$.
The sequence $\vec{a}$ is called the {\em vector\/} of partition $P$.

In words, each partition of $\db_{i}$ groups facts that can be obtained from $F_{i}$ or $G_{i}$ by replacing the variables of $\vec{x}_{i}$ with the same fixed constants.
For example, the attack graph in Fig.~\ref{fig:threecycles} contains an attack cycle involving
$R_{3}(\underline{x,y,u_{3}},u_{4})$ and $R_{4}(\underline{x,y,u_{4}},u_{3})$.
The sequence $\tuple{x,y}$ contains the variables that also occur in other cycles.
The facts $R_{3}(\underline{a,b,c},d)$ and $R_{4}(\underline{a,b,e},f)$ both belong to the partition with vector $\tuple{a,b}$. 

Clearly, two facts that belong to distinct partitions of $\db_{i}$ cannot be key-equal.
It follows that each repair of $\db_{i}$ is a disjoint union of repairs, one for each partition of $\db_{i}$.

Let $\clean{\db_{i}}$ be the smallest subset of $\db_{i}$ that contains every partition $P$ satisfying $P\in\cqak{q_{i}}$. 
By Lemma~\ref{lem:sharedvariables} and~\cite[Theorem~2]{DBLP:journals/ipl/KolaitisP12},
$\cqak{q_{i}}$ is in \PTIME\ for $1\leq i\leq\ell$. 
From the following sublemma, it follows that $\cqak{q}$ is in \PTIME.

\begin{sublemma}\label{prop:shared}
The following are equivalent:
\begin{enumerate}
\item
$\db\in\cqak{q}$.
\item
$\bigcup_{1\leq i\leq\ell}\clean{\db_{i}}\models q$.
\end{enumerate}
\end{sublemma}

For the step of the induction, assume that $F$ is an unattacked atom in $q$'s attack graph.
Let $\vec{x}$ be a sequence of distinct variables such that $\sequencevars{\vec{x}}=\key{F}$.
By Corollary~8.11 in~\cite{DBLP:journals/tods/Wijsen12}, the following are equivalent.
\begin{enumerate}
\item
$\db\in\cqak{q}$.
\item
For some $\vec{a}\in D^{\card{\vec{x}}}$, $\db\in\cqak{\substitute{q}{\vec{x}}{\vec{a}}}$.
\end{enumerate}

Let $\vec{y}$ be a sequence of distinct variables such that $\sequencevars{\vec{y}}=\all{F}\setminus\key{F}$.
Let $q'=q\setminus\{F\}$.
By Lemma~\ref{lem:purified},
it is possible to compute in polynomial time a database $\db'$ that is purified relative to $\substitute{q}{\vec{x}}{\vec{a}}$ such that
$$\db\in\cqak{\substitute{q}{\vec{x}}{\vec{a}}}\iff\db'\in\cqak{\substitute{q}{\vec{x}}{\vec{a}}}.$$
By Lemma~\ref{lem:SE3}, the following are equivalent:
\begin{enumerate}
\item
$\db'\in\cqak{\substitute{q}{\vec{x}}{\vec{a}}}$.
\item
$\db'\neq\emptyset$ and for all $\vec{b}\in D^{\card{\vec{y}}}$,
if $\substitute{F}{\vec{x}\vec{y}}{\vec{a}\vec{b}}\in\db'$,
then $\db'\in\cqak{\substitute{q'}{\vec{x}\vec{y}}{\vec{a}\vec{b}}}$.
\end{enumerate}
By Lemma~\ref{lem:weakweak}, all cycles in the attack graph of $\substitute{q'}{\vec{x}\vec{y}}{\vec{a}\vec{b}}$ are weak and terminal.
By the induction hypothesis, it follows that $\cqak{\substitute{q'}{\vec{x}\vec{y}}{\vec{a}\vec{b}}}$ is in \PTIME.
Since the sizes of $D^{\card{\vec{x}}}$ and $D^{\card{\vec{y}}}$ are polynomially bounded in the size of $\db$,
it is correct to conclude that $\cqak{q}$ is in \PTIME.
\end{proof}

\subsection{Nonterminal Weak Cycles}

Theorems~\ref{the:strongcycle} and~\ref{the:outdegree} leave open the complexity of $\cqak{q}$ when the attack graph of $q$ contains one or more nonterminal weak cycles and no strong cycle.
In this section, we zoom in on acyclic queries $\acq{k}$, defined next for $k\in\{2,3,\dots\}$, whose attack graph contains $\frac{k(k-1)}{2}$ nonterminal weak cycles and no strong cycle.
By showing tractability of $\cqak{\acq{k}}$, we obtain more supporting evidence for Conjecture~\ref{con:weak}.
As a side result, we will solve a complexity issue raised by Fuxman and Miller~\cite{FuxmanM07}. 


\begin{definition}
For $k\geq 2$, let $\cq{k}$ and $\acq{k}$ denote the following Boolean conjunctive queries without self-join.
\begin{eqnarray*}
\cq{k}  & = & \{R_{1}(\underline{x_{1}},x_{2}), R_{2}(\underline{x_{2}},x_{3}), \dots,
                R_{k-1}(\underline{x_{k-1}},x_{k}),\\
        &   &   \phantom{\{}R_{k}(\underline{x_{k}},x_{1})\},\\
\acq{k} & = & \{R_{1}(\underline{x_{1}},x_{2}), R_{2}(\underline{x_{2}},x_{3}), \dots,
                R_{k-1}(\underline{x_{k-1}},x_{k}),\\ 
        &   &   \phantom{\{}R_{k}(\underline{x_{k}},x_{1}), S_{k}(\underline{x_{1},x_{2},\dots,x_{k}})\},
\end{eqnarray*}
where $x_{1},\dots,x_{k}$ are distinct variables and $R_{1},\dots,R_{k},S_{k}$ distinct relation names.
For $i\in\{1,\dots,k\}$, relation name $R_{i}$ is of signature~$\signature{2}{1}$, and $S_{k}$ is of signature~$\signature{k}{k}$.
\end{definition}

Obviously, a query $q$ is acyclic if it contains an atom $F$ such that $\all{F}=\queryvars{q}$.
Therefore, $\acq{k}$ is acyclic because the $S_{k}$-atom contains all variables that occur in the query.
On the other hand, $\cq{k}$ is acyclic if $k=2$ and cyclic if $k\geq 3$. 

For $i\in\{1,\dots,k\}$, the attack graph of $\acq{k}$ contains attacks from the $R_{i}$-atom to every other atom. 
Figure~\ref{fig:acqthree} shows the attack graph of $\acq{3}$.
All attack cycles are weak, but Theorem~\ref{the:outdegree} does not apply because the cycles are nonterminal.

$\cqak{\cq{k}}$ was claimed \coNP-hard for all $k\geq 2$ in~\cite{FuxmanM07}.
Later, however, Wijsen~\cite{Wijsen2010950} found a mistake in the proof of that claim and showed that $\cqak{\cq{k}}$ is tractable if $k=2$. 
The complexity of $\cqak{\cq{k}}$ for $k\geq 3$ will be settled by Corollary~\ref{cor:fuxman}.

\begin{figure}\centering
\includegraphics{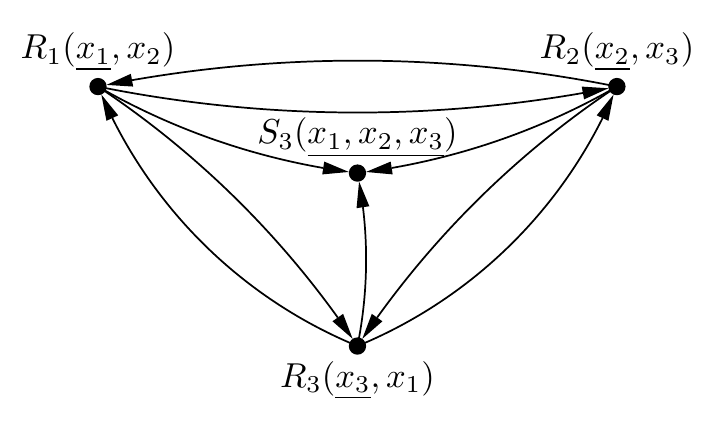}
\caption{Attack graph of $\acq{3}$. All cycles are weak and nonterminal}\label{fig:acqthree}
\end{figure}

\begin{figure}\centering
\includegraphics{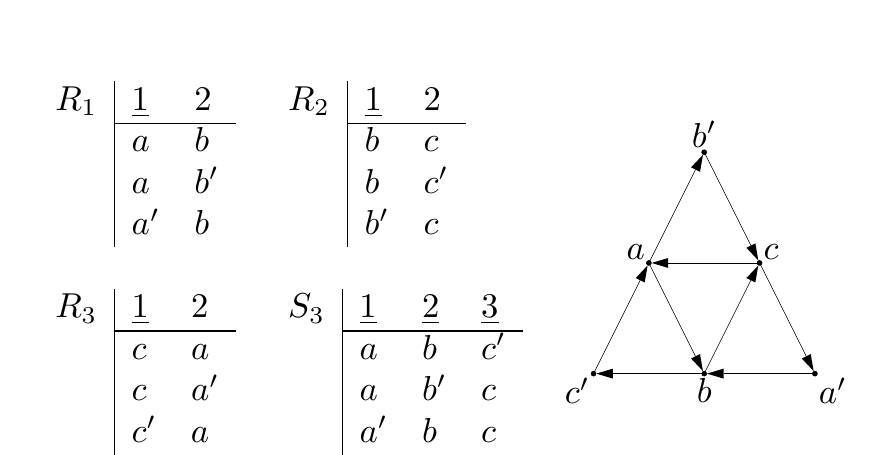}
\caption{At the left: uncertain database that is purified relative to $\acq{3}$.
At the right: graph representation of $R_{1}\cup R_{2}\cup R_{3}$.
Note that the three cycles encoded in $S_{3}$ are clockwise.}\label{fig:acq}
\end{figure}

\begin{figure}\centering
\includegraphics{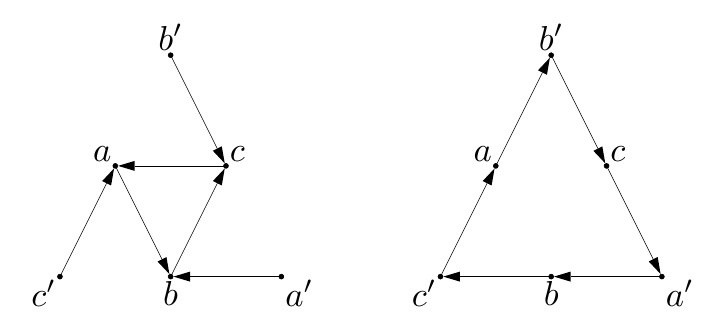}
\caption{Graph representation of two repairs (of the uncertain database of Fig.~\ref{fig:acq}) that falsify $\acq{3}$.
The left cycle is anticlockwise and not encoded in $S_{3}$.
}\label{fig:acqrep}
\end{figure}

\begin{theorem}\label{the:acq}
For $k\geq 2$, $\cqak{\acq{k}}$ is in \PTIME.
\end{theorem}
\begin{proof}
({\em Extended sketch.\/})
Let $\db$ be an uncertain database with schema $\{R_{1}$,\dots, $R_{k}$, $S_{k}\}$.
By Lemma~\ref{lem:purified}, we can assume without loss of generality that $\db$ is purified relative to $\acq{k}$.
Let $D$ be the active domain of $\db$.
For every $i\in\{1,\dots,k\}$, define $\type{x_{i}}$ as the subset of $D$ that contains $a$ if for some valuation $\mu$, $\substitute{\mu}{x_{i}}{a}(\acq{k})\subseteq\db$. 
Since $\acq{k}$ has no self-join, we can assume without loss of generality that $i\neq j$ implies $\type{x_{i}}\cap\type{x_{j}}=\emptyset$.

For example, assume $R_{i}(\underline{a},b),R_{j}(\underline{c},d)\in\db$ with $i<j$.
Since $a\in\type{x_{i}}$, $b\in\type{x_{i+1}}$, and $c\in\type{x_{j}}$,
it follows that $b\neq a\neq c$ and that $b=c$ implies $j=i+1$.

The $R_{i}$-facts of $\db$ can be viewed as edges of a directed graph ($1\leq i\leq k$).
This is illustrated in Fig.~\ref{fig:acq} for $k=3$.
Let $G=(V,E)$ be the directed graph such that $V=D$ and 
$E=\{(a,b)\mid\mbox{$R_{i}(\underline{a},{b})\in\db$ for some $i$}\}$.
Then, $G$ is $k$-partite with vertex classes $\type{x_{1}}$, \dots, $\type{x_{k}}$.
Furthermore, whenever $(a,b)\in E$ and $a\in\type{x_{i}}$, then $b\in\type{x_{i+1}}$ if $i<k$ (and $b\in\type{x_{1}}$ if $i=k$).
It follows that the length of every elementary cycle in $G$ must be a multiple of $k$.
Since $\db$ is purified, no vertex has zero outdegree.
We define $\calC$ as the set of cycles of length $k$ such that if $\db$ contains $S_{k}(\underline{a_{1},\dots,a_{k}})$,
then $\calC$ contains the cycle $a_{1},a_{2},\dots,a_{k},a_{1}$.


Since $\db$ is purified, $G$ is a vertex-disjoint union of strong components $S_{1},\dots,S_{\ell}$ (for some $\ell\geq 0$) such that for $i\neq j$, no edge leads from a vertex in $S_{i}$ to a vertex in $S_{j}$.\footnote{A strong component of a graph $G$ is a maximal strongly connected subgraph of $G$.
A graph is strongly connected if there is a path from any vertex to any other.} 

In what follows, some vertices and edges of $G$ will be {\em marked\/}.
It is straightforward that $\db\not\in\cqak{\acq{k}}$ is equivalent to the following.
\begin{equation}\label{eq:nocalc}
\begin{minipage}{0.42\textwidth}
It is possible to mark exactly one outgoing edge for each vertex of $G$ without marking all edges of some cycle in $\calC$.
\end{minipage}
\end{equation}

We provide a polynomial-time algorithm for testing condition~(\ref{eq:nocalc}).
Marking one outgoing edge for each vertex will create a cycle of marked edges in each strong component.

For each strong component $S_{i}$, consider the following cases successively and execute the first one that applies.

\myparagraph{Case $S_{i}$ contains a cycle of length $k$ that does not belong to $\calC$.}
Such a cycle is illustrated by Fig.~\ref{fig:acqrep} (left).
Mark all vertices and edges of the cycle.
Notice that the number of cycles of length $k$ is at most $\card{V}^{k}$,
which is polynomial in the size of $\db$.

\myparagraph{Case $S_{i}$ contains an elementary cycle of length (strictly) greater than $k$.}
Such a cycle is illustrated by Fig.~\ref{fig:acqrep} (right).
Mark all vertices and edges of the cycle.
To see that this step is in polynomial time, notice that the following are equivalent:
\begin{itemize}
\item
$S_{i}$ contains an elementary cycle of length greater than $k$.
\item
$S_{i}$ contains a path $a_{1},a_{2},\dots,a_{k},a_{k+1}$ such that $a_{1}\neq a_{k+1}$ and $S_{i}$ contains a path from $a_{k+1}$ to $a_{1}$ that contains no edge from $\{a_{1},a_{2},\dots,a_{k}\}\times V$.
\end{itemize}
The latter condition can be tested in polynomial time, because there are at most $\card{V}^{k+1}$ distinct choices for $a_{1},a_{2},\dots,a_{k},a_{k+1}$ and paths can be found in polynomial time.

\myparagraph{Case neither of the above two cases applies.}
Conclude that~(\ref{eq:nocalc}) is false.

\myparagraph{}
If after the previous step every strong component contains a cycle of marked edges, 
then it is correct to conclude that~(\ref{eq:nocalc}) is true.
Notice that every cycle of $\calC$ now contains at least one unmarked edge.
We can  achieve~(\ref{eq:nocalc}) by marking, for each yet unmarked vertex, the vertices and edges on a shortest path to some marked vertex.
This can be done without creating new cycles of marked edges.
\end{proof}

Since query $\cq{k}$ is acyclic if $k\geq 3$, attack graphs are not defined for $\cq{k}$ if $k\geq 3$.
Nevertheless, the following lemma immediately implies that if $\cqak{\acq{k}}$ is tractable,
then so is $\cqak{\cq{k}}$.

\begin{lemma}\label{lem:cq}
Let $q$ be a Boolean conjunctive query without self-join.
If $q'\subseteq q$ and every atom in $q\setminus q'$ is all-key,
then there exists an \ACO\ many-one reduction from $\cqak{q'}$ to $\cqak{q}$. 
\end{lemma}

\begin{corollary}\label{cor:fuxman}
For $k\geq 2$, $\cqak{\cq{k}}$ is in \PTIME.
\end{corollary}

Unsurprisingly, there exist acyclic queries $q\not\in\{\acq{k}\mid k\geq 2\}$ whose attack graph contains some nonterminal cycle and no strong cycle.
The complexity of $\cqak{q}$ for such queries $q$ is open.


\section{Uncertainty and Probability}\label{sec:bid}

In this section, we study the relationship between the complexities of $\cqak{q}$ and evaluating $q$ on probabilistic databases.
The motivation is that, on input of an uncertain database $\db$, the problem $\cqak{q}$ is solved if we can determine whether query $q$ evaluates to probability~$1$ on the probabilistic database obtained from $\db$ by assuming a uniform probability distribution over the set of repairs of $\db$. 
We show, however, that this approach provides no new insights in the tractability frontier of $\cqak{q}$.

\subsection{Background from Probabilistic Databases}

In this section, we review an important result from probabilistic database theory.

\begin{definition}
A {\em possible world\/} $\pw$ of uncertain database $\db$ is a consistent subset of $\db$.
The set of possible worlds of $\db$ is denoted $\pworlds{\db}$.
Notice that possible worlds, unlike repairs, need not be maximal consistent.

A {\em probabilistic database\/} is a pair $(\db,\probsymbol)$ where $\db$ is an uncertain database and  
$\probsymbol:\pworlds{\db}\rightarrow\probset$ is a total function such that 
$\sum_{\pw\in\pworlds{\db}}\prob{\pw}=1$.
We will assume that the numbers in the codomain of $\probsymbol$ are rational.
\end{definition}

The following definition extends the function $\probsymbol$ to Boolean first-order queries $q$.

\begin{definition}\label{def:extpr}
Let $(\db,\probsymbol)$ be a probabilistic database.
Let $q$ be a Boolean first-order query.
We define 
$$\prob{q}=\sum_{\pw\in\pworlds{\db}:\pw\models q}\prob{\pw}.$$
In words, $\prob{q}$ sums up the probabilities of the possible worlds that satisfy $q$.
\end{definition}

Of special interest is the application of Definition~\ref{def:extpr} in case $q$ is a single fact, or a Boolean combination of facts.
Notice that if $(\db,\probsymbol)$ is a probabilistic database and $A_{1},\dots,A_{n}$ are distinct facts belonging to a same block of $\db$, 
then $\prob{A_{1}\lor A_{2}\lor\dots\lor A_{n}}=\sum_{i=1}^{n}\prob{A_{i}}$, because no possible world can contain two distinct facts that belong to a same block.

\begin{definition}
Probabilistic database $(\db,\probsymbol)$ is called {\em block-independent-disjoint\/} (BID) if the following holds:
whenever $A_{1},\dots,A_{n}$ are facts of $\db$ taken from $n$ distinct blocks (for some $n\geq 1$), then 
$\prob{A_{1}\land A_{2}\land\dots\land A_{n}}=\prod_{i=1}^{n}\prob{A_{i}}$.
\end{definition}

Theorem~2.4 in~\cite{DBLP:journals/jcss/DalviRS11} implies that every BID probabilistic database $(\db,\probsymbol)$ is uniquely determined if $\prob{A}$ is given for every fact $A\in\db$.
This allows for an efficient encoding: rather than specifying $\prob{\pw}$ for every $\pw\in\pworlds{\db}$,
it suffices to specify $\prob{A}$ for every $A\in\db$.
In the complexity results that follow, this efficient encoding is assumed.

Notice that we can turn an uncertain database $\db$ into a BID probabilistic database by assuming that the probabilities of all repairs are equal and sum up to $1$.
A consistent subset of $\db$ that is not maximal, would then have zero probability.

\begin{function}
\KwIn{
$q$ is a Boolean conjunctive query without self-join.}
\KwResult{
Boolean in $\{\waar,\onwaar\}$.}
\Begin{
\SE{1}\If{$\card{q}=1$ and $\queryvars{q}=\emptyset$}
{\Return{$\waar$}\;}

\SE{2}\If{$q=q_{1}\cup q_{2}$ with $q_{1}\neq\emptyset\neq q_{2}$, $\queryvars{q_{1}}\cap\queryvars{q_{2}}=\emptyset$}
{\Return{$\ref{algo:issafe}(q_{1})\land\ref{algo:issafe}(q_{2})$}\;}

\tcc{$a$ is an arbitrary constant}
\SE{3}\If{$\bigcap_{F\in q}\key{F}\neq\emptyset$}
{select $x\in\bigcap_{F\in q}\key{F}$\;
\Return{$\ref{algo:issafe}(\substitute{q}{x}{a})$}\;}

\SE{4}\If{there exists $F\in q$ such that $\key{F}=\emptyset\neq\all{F}$}
{select $F\in q$ such that $\key{F}=\emptyset\neq\all{F}$;
 select $x\in\all{F}$\;
\Return{$\ref{algo:issafe}(\substitute{q}{x}{a})$}\;}

\If{none of the above}
{\Return{$\onwaar$}\;}
}
\caption{IsSafe($q$) Determine whether $q$ is safe~\cite{DBLP:journals/jcss/DalviRS11}}\label{algo:issafe}
\end{function}

\begin{definition}
Given a Boolean query $q$, $\prok{q}$ is the following function problem:
on input of a BID probabilistic database $(\db,\probsymbol)$, determine the value of $\prob{q}$.

A Boolean conjunctive query $q$, without self-join, is called {\em safe\/} if Algorithm~\ref{algo:issafe} returns $\waar$.
\end{definition}

The following result establishes a dichotomy in the complexity of $\prok{q}$.

\begin{theorem}[\cite{DBLP:journals/jcss/DalviRS11}]\label{the:suciu}
Let $q$ be a Boolean conjunctive query without self-join.
\begin{enumerate}
\item
If $q$ is safe, then $\prok{q}$ is in \FPTIME.
\item
If $q$ is not safe, then $\prok{q}$ is \sharpP-hard.
\end{enumerate}
\end{theorem}

\subsection{Comparing Complexities}

The following proposition establishes a straightforward relationship between the problems $\prok{q}$ and $\cqak{q}$.
The only subtlety is that a repair contains a fact of each block,
while a possible world and a block may have an empty intersection (recall that possible worlds, unlike repairs, need not be maximal consistent).
In the statement of this proposition,
$\db'$ restricts $\db$ to the set of blocks whose probabilities sum up to~$1$. 
 
\begin{proposition}\label{prop:bid}
Let $(\db,\probsymbol)$ be a BID probabilistic database.
Let $\db'$ be the smallest subset of $\db$ that contains every block $\block$ of $\db$ such that $\sum_{A\in\block}\prob{A}=1$.
Let $q$ be a Boolean conjunctive query.
Then the following are equivalent:
\begin{enumerate}
\item
$\db'\in\cqak{q}$.
\item
On input $(\db,\probsymbol)$, the answer to the function problem $\prok{q}$ is $1$.
\end{enumerate}
\end{proposition}

The following theorem establishes a nontrivial relationship between the complexities of $\cqak{q}$ and $\prok{q}$. 
Notice that the query $q$ in the theorem's statement is not required to be acyclic.

\begin{theorem}\label{the:prob}
Let $q$ be a Boolean conjunctive query without self-join.
If $q$ is safe, then $\cqak{q}$ is first-order expressible. 
\end{theorem}
\begin{proof}
The proof runs by induction on the execution of Algorithm~\ref{algo:issafe}.
Since $q$ is safe, some rule of \ref{algo:issafe} applies to $q$.

\myparagraph{Case \se{1} applies.} 
If $q$ consists of a single fact, then $\cqak{q}$ is obviously first-order expressible.

\myparagraph{Case \se{2} applies.} 
Let $q=q_{1}\cup q_{2}$ with $q_{1}\neq\emptyset\neq q_{2}$ and $\queryvars{q_{1}}\cap\queryvars{q_{2}}=\emptyset$.
Since $q$ is safe, $q_{1}$ and $q_{2}$ are safe by definition of safety.
By the induction hypothesis,
there exists a certain first-order rewriting $\varphi_{1}$ of $q_{1}$, and a certain first-order rewriting $\varphi_{2}$ of $q_{2}$.
Obviously, $\varphi_{1}\land\varphi_{2}$ is a certain first-order rewriting of $q$.

\myparagraph{Case \se{3} applies.} 
Assume variable $x$ such that for every $F\in q$, $x\in\key{F}$.
By definition of safety, $\substitute{q}{x}{a}$ is safe.
It can be easily seen that $\db\in\cqak{q}$ if and only if for some constant $a$,
$\db\in\cqak{\substitute{q}{x}{a}}$.
By the induction hypothesis, $\cqak{\substitute{q}{x}{a}}$ is first-order expressible.
Let $\varphi$ be a certain first-order rewriting of $\substitute{q}{x}{c}$, where we assume without loss of generality that $c$ is a constant that does not occur in $q$.
Let $\varphi(x)$ be the first order formula obtained from $\varphi$ by replacing each occurrence of $c$ with $x$.
Then, $\exists x\varphi(x)$ of a certain first-order rewriting of $q$.

\myparagraph{Case \se{4} applies.} 
Assume $F\in q$ such that $\key{F}=\emptyset$ and $\all{F}\neq\emptyset$.
Let $\vec{x}$ be a sequence of distinct variables such that $\sequencevars{\vec{x}}=\all{F}$. 
Let $\vec{a}=\tuple{a,a,\dots,a}$ be a sequence of length $\card{\vec{x}}$.
By definition of safety, $\substitute{q}{\vec{x}}{\vec{a}}$ is safe.
By the induction hypothesis, $\cqak{\substitute{q}{\vec{x}}{\vec{a}}}$ is first-order expressible.
From Lemma~8.6 in~\cite{DBLP:journals/tods/Wijsen12}, it follows that $\cqak{q}$ is first-order expressible
\end{proof}

\begin{corollary}\label{cor:prob}
Let $q$ be a Boolean conjunctive query without self-join.
If $\cqak{q}$ is not first-order expressible,
then the function problem $\prok{q}$ is \sharpP-hard.
\end{corollary}

For acyclic queries, the only complexities of $\cqak{q}$ left open by Theorems~\ref{the:acyclic}, \ref{the:strongcycle}, and~\ref{the:outdegree} concern queries $q$ with a cyclic attack graph (in particular, an attack graph without strong cycle and with at least one nonterminal weak cycle).
For such a query $q$, $\cqak{q}$ is not first-order expressible (by Theorem~\ref{the:acyclic}), hence $\prok{q}$ is \sharpP-hard (by Corollary~\ref{cor:prob}).
Consequently, the probabilistic database approach fails to provide further insight into the tractability frontier of $\cqak{q}$.
It turns out that the queries $q$ for which $\prok{q}$ is tractable is a very restricted subset of the queries for which $\cqak{q}$ is tractable (assuming $\FPTIME\neq\sharpP$ and $\PTIME\neq\coNP$).

\section{Discussion}\label{sec:discussion}

In the following, we say that a class $\calP$ of function problems {\em exhibits an effective \FPTIME-\sharpP-dichotomy\/}
if all problems in $\calP$ are either in \FPTIME\ or \sharpP-hard and it is decidable whether a given problem in $\calP$ is in \FPTIME\ or \sharpP-hard.
Likewise, we say that a class $\calP$ of decision problems {\em exhibits an effective \PTIME-\coNP-dichotomy\/}
if all problems in $\calP$ are either in \PTIME\ or \coNP-hard and it is decidable whether a given problem in $\calP$ is in \PTIME\ or \coNP-hard.

Recall from Section~\ref{sec:related} that $\sharpCQA{q}$ is the counting variant of $\cqak{q}$,
which takes as input an uncertain database $\db$ and asks how many repairs of $\db$ satisfy query $q$.
For the probabilistic and counting variants of $\cqak{q}$, the following dichotomies have been established.

\begin{theorem}[\cite{DBLP:journals/jcss/DalviRS11},\cite{MASLOWSKIJCSS2012}]\label{the:dichotomy}
The following classes exhibit an effective \FPTIME-\sharpP-dichotomy:
\begin{enumerate}
\item
the class containing $\prok{q}$ for all Boolean conjunctive queries $q$ without self-join; and
\item
the class containing $\sharpCQA{q}$ for all Boolean conjunctive queries $q$ without self-join.
\end{enumerate}
\end{theorem}

Theorem~\ref{the:strongcycle} and Conjecture~\ref{con:weak} imply the following conjecture, which is thus weaker than Conjecture~\ref{con:weak}.

\begin{conjecture}\label{con:dichotomy}
The class containing $\cqak{q}$ for all acyclic Boolean conjunctive queries $q$ without self-join exhibits an effective \PTIME-\coNP-dichotomy.
\end{conjecture}
 
From Theorems~\ref{the:strongcycle} and~\ref{the:outdegree}, it follows that in order to prove Conjecture~\ref{con:dichotomy}, it suffices to show that an effective \PTIME-\coNP-dichotomy is exhibited by the class containing $\cqak{q}$ for all queries $q$ whose attack graph contains some nonterminal cycle and no strong cycle.

We confidently believe that the \mbox{\PTIME-\coNP}-dichotomy of Conjecture~\ref{con:dichotomy} (if true) will be harder to prove than the \mbox{\FPTIME-\sharpP}-dichotomies established by Theorem~\ref{the:dichotomy}, for the following reasons. 
All problems $\prok{q}$ that are in \FPTIME\ can be solved by a single, fairly simple polynomial-time algorithm which appears in~\cite{DBLP:journals/jcss/DalviRS11}.
Likewise, all problems $\sharpCQA{q}$ in \FPTIME\ can be solved by a single, fairly simple polynomial-time algorithm~\cite{MASLOWSKIJCSS2012}.
On the other hand, $\cqak{q}$ problems in \PTIME\ seem to ask for sophisticated polynomial-time algorithms.
In their proof that Conjecture~\ref{con:dichotomy} holds for queries with exactly two atoms,
Kolaitis and Pema~\cite{DBLP:journals/ipl/KolaitisP12} made use of an ingenious polynomial-time algorithm of Minty~\cite{DBLP:journals/jct/Minty80}.
Our proof of Theorem~\ref{the:acq} uses algorithms from (directed) graph theory. 
Despite their sophistication, these polynomial-time algorithms only solve restricted cases of $\cqak{q}$. 

Notice also that by Corollary~\ref{cor:prob} and Theorem~\ref{the:acyclic}, the function problem $\prok{q}$ is intractable for all acyclic queries $q$ with a cyclic attack graph.
On the other hand, cycles in attack graphs are exactly what makes Conjecture~\ref{con:dichotomy} hard to prove.

\begin{spacing}{0.9}

\end{spacing}

\appendix

\section{Proofs of Section~\ref{sec:preliminaries}}

\begin{proof}[Lemma~\ref{lem:purified}]
In polynomial time, we can construct a maximal sequence
$\db_{0}\stackrel{A_{1}}{\longrightarrow}\db_{1}\stackrel{A_{2}}{\longrightarrow}\db_{2}\dots\stackrel{A_{n}}{\longrightarrow}\db_{n}$ such that for $i\in\{1,2,\dots,n\}$,
\begin{enumerate}
\item
$A_{i}\in\db_{i-1}$;
\item
there exists no valuation $\theta$ such that $A_{i}\in\theta(q)\subseteq\db_{i-1}$; and 
\item
$\db_{i}=\db_{i-1}\setminus\kc{A_{i}}{\db_{i-1}}$.
\end{enumerate}
Clearly, $\db_{n}$ is purified relative to $q$.
We show that for $i\in\{1,2,\dots,n\}$,
$$\db_{i-1}\in\cqak{q}\iff\db_{i}\in\cqak{q}.$$
\framebox{$\implies$} 
By contraposition.
Assume that $\rep$ is a repair of $\db_{i}$ such that $\rep\not\models q$.
Then, $\rep\cup\{A_{i}\}$ is a repair of $\db_{i-1}$ that falsifies $q$.
\framebox{$\impliedby$}
By contraposition.
Assume that $\rep$ is a repair of $\db_{i-1}$ such that $\rep\not\models q$. 
Obviously, $\rep\setminus\kc{A_{i}}{\db_{i-1}}$ is a repair of $\db_{i}$ that falsifies $q$.
\end{proof}

\section{Proofs of Section~\ref{sec:attackgraph}}

\begin{proof}[Lemma~\ref{lem:rsu}]
Assume $\attacksq{F}{G}{q}$.
Let $\tau$ be a join tree for $q$.
Let $F\step{L_{1}}F_{1}\dots\step{L_{m-1}}F_{m-1}\step{L_{m}}G$ be the path in $\tau$ between $F$ and $G$ ($m\geq 1$).

We have $L_{m}\subseteq\all{G}$.
Since $\FD{q\setminus\{F\}}\models\fd{\key{G}}{L_{m}}$ and $L_{m}\nsubseteq\keycl{F}{q}$ (because $\attacksq{F}{G}{q}$),
it must be the case that $\FD{q\setminus\{F\}}\not\models\fd{\key{F}}{\key{G}}$,
hence $\key{G}\nsubseteq\keycl{F}{q}$.

We have $L_{1}\subseteq\all{F}$.
Since $L_{1}\nsubseteq\keycl{F}{q}$ (because $\attacksq{F}{G}{q}$),
it must be the case that $\all{F}\nsubseteq\keycl{F}{q}$. 
\end{proof}

\section{Proofs of Section~\ref{sec:intractability}}

\begin{proof}[Lemma~\ref{lem:two}]
We show that if the attack graph of $q$ contains a strong cycle of length $n$ with $n\geq 3$,
then it contains a strong cycle of some length $m$ with $m<n$.

Let $H_{0}\attacksymbolq{q}H_{1}\attacksymbolq{q}H_{2}\dotsc\attacksymbolq{q}H_{n-1}\attacksymbolq{q}H_{0}$ be a strong cycle of length $n$ ($n\geq 3$) in the attack graph of $q$, where $i\neq j$ implies $H_{i}\neq H_{j}$.
Assume without loss of generality that the attack $\attacksq{H_{0}}{H_{1}}{q}$ is strong.
Thus, $\FD{q}\not\models\fd{\key{H_{0}}}{\key{H_{1}}}$.

We write $i\oplus j$ as shorthand for for $(i+j)\mod n$.
If $\attacksq{H_{1}}{H_{1\oplus 2}}{q}$,
then $H_{0}\attacksymbolq{q}H_{1}\attacksymbolq{q}H_{1\oplus 2}\dotsc\attacksymbolq{q}H_{n-1}\attacksymbolq{q}H_{0}$ is a strong cycle of length $n-1$, and the desired result holds.
Assume next $\nattacksq{H_{1}}{H_{1\oplus 2}}{q}$.
By Lemma~\ref{lem:A2}, $\attacksq{H_{2}}{H_{1}}{q}$.
We distinguish two cases.

\myparagraph{Case $\attacksq{H_{2}}{H_{1}}{q}$ is a strong attack.}
Then $H_{1}\attacksymbolq{q}H_{2}\attacksymbolq{q}H_{1}$ is a strong cycle of length $2<n$.

\myparagraph{Case $\attacksq{H_{2}}{H_{1}}{q}$ is a weak attack.}
If $\attacksq{H_{1}}{H_{0}}{q}$, then $H_{0}\attacksymbolq{q}H_{1}\attacksymbolq{q}H_{0}$ is a strong cycle of length $2<n$.
Assume next $\nattacksq{H_{1}}{H_{0}}{q}$.
Then, from $H_{0}\attacksymbolq{q}H_{1}\attacksymbolq{q}H_{2}$ and Lemma~\ref{lem:A2},
it follows $\attacksq{H_{0}}{H_{2}}{q}$.
The cycle $H_{0}\attacksymbolq{q}H_{2}\attacksymbolq{q}H_{2\oplus 1}\dotsc\attacksymbolq{q}H_{n-1}\attacksymbolq{q}H_{0}$ has length $n-1$.
It suffices to show that the attack $H_{0}\attacksymbolq{q}H_{2}$ is strong.
Assume towards a contradiction that the attack $H_{0}\attacksymbolq{q}H_{2}$ is weak.
Then, $\FD{q}\models\fd{\key{H_{0}}}{\key{H_{2}}}$.
Since $\attacksq{H_{2}}{H_{1}}{q}$ is a weak attack, $\FD{q}\models\fd{\key{H_{2}}}{\key{H_{1}}}$.
By transitivity, $\FD{q}\models\fd{\key{H_{0}}}{\key{H_{1}}}$, a contradiction.
This concludes the proof.
\end{proof}

\begin{subproof}[Sublemma~\ref{prop:F}]
\framebox{1.$\implies$}
Obviously, $\key{F}\subseteq\keycl{F}{q}$.
From $\attacksq{G}{F}{q}$ and Lemma~\ref{lem:rsu}, it follows $\key{F}\nsubseteq\keycl{G}{q}$.
Thus, we can assume $u\in\key{F}$ such that $u\in\keycl{F}{q}\setminus\keycl{G}{q}$.
Hence, $\rv{\theta_{1}}(u)=\theta_{1}(x)$ and $\rv{\theta_{2}}(u)=\theta_{2}(x)$.
Since $\rv{\theta_{1}}(F)$ and $\rv{\theta_{2}}(F)$ are key-equal by the premise,
$\rv{\theta_{1}}$ and $\rv{\theta_{2}}$ agree on all variables of $\key{F}$.
In particular, $\rv{\theta_{1}}(u)=\rv{\theta_{2}}(u)$.
It follows $\theta_{1}(x)=\theta_{2}(x)$.
\framebox{1.$\impliedby$}
Assume $\theta_{1}(x)=\theta_{2}(x)$.
Since $\key{F}\subseteq\keycl{F}{q}$,
and since neither $y$ nor $z$ occurs inside $\keycl{F}{q}$ in the Venn diagram of Fig.~\ref{fig:schema},
it is correct to conclude that $\rv{\theta_{1}}(F)$ and $\rv{\theta_{2}}(F)$ are key-equal.

\framebox{2.$\implies$}
From $\attacksq{F}{G}{q}$ and Lemma~\ref{lem:rsu}, it follows $\all{F}\nsubseteq\keycl{F}{q}$.
Since $\key{F}\subseteq\keycl{F}{q}$, we can assume a variable $u\in\all{F}\setminus\key{F}$ such that $u\not\in\keycl{F}{q}$.
From the premise $\rv{\theta_{1}}(F)=\rv{\theta_{2}}(F)$, it follows $\rv{\theta_{1}}(u)=\rv{\theta_{2}}(u)$.
Since $y$ occurs in all regions outside $\keycl{F}{q}$ in the Venn diagram, we conclude $\theta_{1}(y)=\theta_{2}(y)$.
Finally, $\theta_{1}(x)=\theta_{2}(x)$ follows from item~$1$ proved before.
\framebox{2.$\impliedby$}
Assume $\theta_{1}(x)=\theta_{2}(x)$ and $\theta_{1}(y)=\theta_{2}(y)$.
Since $\all{F}\subseteq\keyclsup{F}{q}$,
and since $z$ does not occur inside $\keyclsup{F}{q}$ in the Venn diagram,
it is correct to conclude $\rv{\theta_{1}}(F)=\rv{\theta_{2}}(F)$.
This concludes the proof of Sublemma~\ref{prop:F}.
\end{subproof}

\begin{subproof}[Sublemma~\ref{prop:G}]
\framebox{1.$\implies$}
Obviously, $\key{G}\subseteq\keycl{G}{q}$.
Since $\attacksq{F}{G}{q}$ is a strong attack,
$\key{G}\nsubseteq\keyclsup{F}{q}$.
We can assume $u\in\key{G}$ such that $u\in\keycl{G}{q}\setminus\keyclsup{F}{q}$. 
Consequently, $\rv{\theta_{1}}(u)=\skolem{\theta_{1}(y)}{\theta_{1}(z)}$ and $\rv{\theta_{2}}(u)=\skolem{\theta_{2}(y)}{\theta_{2}(z)}$.
Since $\rv{\theta_{1}}(F)$ and $\rv{\theta_{2}}(F)$ are key-equal by the premise,
$\rv{\theta_{1}}$ and $\rv{\theta_{2}}$ agree on all variables of $\key{G}$.
In particular, $\rv{\theta_{1}}(u)=\rv{\theta_{2}}(u)$.
It follows $\theta_{1}(y)=\theta_{2}(y)$ and $\theta_{1}(z)=\theta_{2}(z)$.
\framebox{1.$\impliedby$}
Assume $\theta_{1}(y)=\theta_{2}(y)$ and $\theta_{1}(z)=\theta_{2}(z)$.
Since $\key{G}\subseteq\keycl{G}{q}$,
and since $x$ does not occur inside $\keycl{G}{q}$ in the Venn diagram,
it is correct to conclude that $\rv{\theta_{1}}(G)$ and $\rv{\theta_{2}}(G)$ are key-equal.

\framebox{2.$\implies$}
From $\attacksq{G}{F}{q}$ and Lemma~\ref{lem:rsu}, it follows $\all{G}\nsubseteq\keycl{G}{q}$.
Since $\key{G}\subseteq\keycl{G}{q}$, we can assume a variable $u\in\all{G}\setminus\key{G}$ such that $u\not\in\keycl{G}{q}$.
From the premise $\rv{\theta_{1}}(G)=\rv{\theta_{2}}(G)$, it follows $\rv{\theta_{1}}(u)=\rv{\theta_{2}}(u)$.
Since $x$ occurs in all regions outside $\keycl{G}{q}$ in the Venn diagram,
we conclude $\theta_{1}(x)=\theta_{2}(x)$.
Finally, $\theta_{1}(y)=\theta_{2}(y)$ and $\theta_{1}(z)=\theta_{2}(z)$ follow from item~$1$ proved before.
\framebox{2.$\impliedby$}
Trivial.
This concludes the proof of Sublemma~\ref{prop:G}.
\end{subproof}

\begin{subproof}[Sublemma~\ref{prop:bijection}]
\framebox{1.}
Assume $\rep_{0}$ is a repair of $\db_{0}$.
We first show that $\transform{\rep_{0}}\cap\db_{F}$ contains no two distinct key-equal facts.
Let $A,B\in\transform{\rep_{0}}\cap\db_{F}$.
We can assume $\theta_{1},\theta_{2}\in\calV$ such that $\theta_{1}(F_{0}),\theta_{2}(F_{0})\in\rep_{0}$, $A=\rv{\theta_{1}}(F)$, and $B=\rv{\theta_{2}}(F)$.
By Sublemma~\ref{prop:F}, if $A$ and $B$ are key-equal and distinct,
then $\theta_{1}(F_{0})$ and $\theta_{2}(F_{0})$ are key-equal and distinct, contradicting that $\rep_{0}$ is a repair.
We conclude by contradiction that $\transform{\rep_{0}}\cap\db_{F}$ contains no distinct key-equal facts.

In an analogous way, one can use Sublemma~\ref{prop:G} to show that $\transform{\rep_{0}}\cap\db_{G}$ contains no distinct key-equal facts.
Since $\db_{\rest}$ is consistent, it follows that $\transform{\rep_{0}}$ is consistent.

We next show that $\transform{\rep_{0}}$ is a maximal consistent subset of $\db$.
Let $A\in\db_{F}$.
We need to show that $\transform{\rep_{0}}$ contains a fact that is key-equal to $A$.
We can assume $\theta\in\calV$ such that $A=\rv{\theta}(F)$.
Since $\db_{0}$ contains $\theta(F_{0})$ (by definition of $\calV$),
$\rep_{0}$ contains $\theta'(F_{0})$ for some $\theta'\in\calV$ with $\theta'(x)=\theta(x)$.
Consequently, $\transform{\rep_{0}}$ contains $\rv{\theta'}(F)$.
By Sublemma~\ref{prop:F}, $A$ and $\rv{\theta'}(F)$ are key-equal.
We conclude that $\transform{\rep_{0}}$ contains a fact that is key-equal to $A$.

In an analogous way, one can use Sublemma~\ref{prop:G} to show that for every $A\in\db_{G}$,
$\transform{\rep_{0}}$ contains a fact that is key-equal to $A$.

\framebox{2.}
Let $\rep$ be a repair of $\db$.
Let $\rep_{0}$ be the following subset of $\db_{0}$.
$$
\begin{array}{rccl}
\rep_{0} 
         & = &      & \{\theta(F_{0})\mid\rv{\theta}(F)\in\rep,\theta\in\calV\}\\
         &   & \cup & \{\theta(G_{0})\mid\rv{\theta}(G)\in\rep,\theta\in\calV\}
\end{array}
$$   
We show $\rep=\transform{\rep_{0}}$.
Since $\db_{\rest}\subseteq\rep\cap\transform{\rep_{0}}$,
it suffices to show $\rep\setminus\db_{\rest}\subseteq\transform{\rep_{0}}$ and
$\transform{\rep_{0}}\setminus\db_{\rest}\subseteq\rep$.
Let $A\in\rep\setminus\db_{\rest}$.
Let $A=\rv{\theta}(F)$, $\theta\in\calV$ (the case where $A=\rv{\theta}(G)$ is analogous).
Then by definition of $\rep_{0}$, $\theta(F_{0})\in\rep_{0}$, hence $\rv{\theta}(F)\in\transform{\rep_{0}}$.
Conversely, let $A\in\transform{\rep_{0}}\setminus\db_{\rest}$.
Let $A=\rv{\theta}(F)$, $\theta\in\calV$ (the case where $A=\rv{\theta}(G)$ is analogous).
By~(\ref{eq:F}), $\theta(F_{0})\in\rep_{0}$.
We can assume $\theta'\in\calV$ such that $\rv{\theta'}(F)\in\rep$ and $\theta'(F_{0})=\theta(F_{0})$.
By Sublemma~\ref{prop:F}, $\theta'(F_{0})=\theta(F_{0})$ implies $\rv{\theta'}(F)=\rv{\theta}(F)$,
hence $A\in\rep$.

Using Sublemmas~\ref{prop:F} and~\ref{prop:G}, it is straightforward to show that $\rep_{0}$ is a repair of $\db_{0}$.

\framebox{3.}
Let $\rep_{0},\rep_{0}'$ be distinct repairs of $\db_{0}$.
Then there exist distinct key-equal facts $A,B$ such that $A\in\rep_{0}$ and $B\in\rep_{0}'$. 
Assume $A,B$ are $R_{0}$-facts (the case where $A,B$ are $S_{0}$-facts is analogous).
There exist valuations $\theta_{1},\theta_{2}\in\calV$ such that $A=\theta_{1}(F_{0})$ and $B=\theta_{2}(F_{0})$.
By Sublemma~\ref{prop:F}, $\rv{\theta_{1}}(F)$ and $\rv{\theta_{2}}(F)$ are distinct and key-equal.
Since $\rv{\theta_{1}}(F)\in\transform{\rep_{0}}$ and $\rv{\theta_{2}}(F)\in\transform{\rep_{0}'}$,
and since $\transform{\rep_{0}}$, $\transform{\rep_{0}'}$ are consistent by property~1 shown earlier,
it follows $\transform{\rep_{0}}\neq\transform{\rep_{0}'}$.
This concludes he proof of Sublemma~\ref{prop:bijection}.
\end{subproof}

\section{Proofs of Section~\ref{sec:tractability}}

\begin{proof}[Lemma~\ref{lem:weakweak}]
\framebox{1.}
Let $\tau$ be a join tree for $q$.
Let $\tau'$ be the graph obtained from $\tau$ by replacing each vertex $H$ with $\substitute{H}{z}{c}$,
and by adjusting edge labels (that is, every label $L$ is replaced with $L\setminus\{z\}$).
Clearly, $\tau'$ is a join tree for $q'$.

\framebox{2 and 3.}
Let $Q\subseteq q$.
Let  $X,Y\subseteq\queryvars{Q}$.
Let $Q'=\substitute{Q}{z}{c}$.
In the next paragraph, we show that $\FD{Q}\models\fd{X}{Y}$ implies $\FD{Q'}\models\fd{X\setminus\{z\}}{Y\setminus\{z\}}$.

The computation of the attribute closure $\{y\mid\FD{Q}\models\fd{X}{y}\}$ by means of a standard algorithm~\cite[page~165]{ABITEBOUL95} corresponds to constructing a maximal sequence
$$
\begin{array}{llll}
X  & = & S_{0}   & H_{1}\\
   &   & S_{1}   & H_{2}\\
   &   & \multicolumn{1}{c}{\vdots} & \multicolumn{1}{c}{\vdots}\\
   &   & S_{k-1} & H_{k}\\
   &   & S_{k}
\end{array}
$$
where 
\begin{enumerate}
\item $S_{0}\subsetneq S_{1}\subsetneq\dots\subsetneq S_{k-1}\subsetneq S_{k}$; and
\item for every $i\in\{1,2,\dots,k\}$, 
\begin{enumerate}
\item
$H_{i}\in Q$. Thus, $\FD{Q}$ contains the functional dependency $\fd{\key{H_{i}}}{\all{H_{i}}}$.
\item
$\key{H_{i}}\subseteq S_{i-1}$ and $S_{i}=S_{i-1}\cup\all{H_{i}}$.
\end{enumerate}
\end{enumerate}
Then, $S_{k}=\{y\mid\FD{Q}\models\fd{X}{y}\}$. Consider now the following sequence.
$$
\begin{array}{llll}
X\setminus\{z\}  & = & S_{0}\setminus\{z\}   & \substitute{H_{1}}{z}{c}\\
                 &   & S_{1}\setminus\{z\}   & \substitute{H_{2}}{z}{c}\\
                 &   & \multicolumn{1}{c}{\vdots} & \multicolumn{1}{c}{\vdots}\\
                 &   & S_{k-1}\setminus\{z\} & \substitute{H_{k}}{z}{c}\\
                 &   & S_{k}\setminus\{z\}
\end{array}
$$
Clearly, for every $i\in\{1,2,\dots,k\}$, 
\begin{enumerate}
\item
$\substitute{H_{i}}{z}{c}\in Q'$. Thus, $\FD{Q'}$ contains the functional dependency $\fd{\key{\substitute{H_{i}}{z}{c}}}{\all{\substitute{H_{i}}{z}{c}}}$.
\item
$\key{\substitute{H_{i}}{z}{c}}\subseteq S_{i-1}\setminus\{z\}$ and $S_{i}\setminus\{z\}=\bigformula{S_{i-1}\setminus\{z\}}\cup\all{\substitute{H_{i}}{z}{c}}$.
\end{enumerate}
It follows that if $\FD{Q}\models\fd{X}{y}$ and $y\neq z$, then $\FD{Q'}\models\fd{X\setminus\{z\}}{y}$.

To prove~2, assume $\nattacks{F}{G}{q}$.
Then, the unique path in $\tau$ between $F$ and $G$ contains an edge with label $L$ such that
$\FD{q\setminus\{F\}}\models\fd{\key{F}}{L}$.
It follows  $\FD{q'\setminus\{F'\}}\models\fd{\key{F'}}{L\setminus\{z\}}$.
Since $L\setminus\{z\}$ is a label on the unique path in $\tau'$ between $F'$ and $G'$,
it follows $\nattacks{F'}{G'}{q'}$.

To prove~3, assume the attack $\attacks{F}{G}{q}$ is weak.
Then $\FD{q}\models\fd{\key{F}}{\key{G}}$, hence $\FD{q'}\models\fd{\key{F'}}{\key{G'}}$.
It follows that the attack $\attacks{F'}{G'}{q'}$, if it exists, is weak.
\end{proof}

\begin{proof}[Lemma~\ref{lem:terminalcycles}]
Assume each cycle in the attack graph of $q$ is terminal.
Assume towards a contradiction that three distinct atoms $F,G,H$ of $q$ belong to a same cycle such that $\attacksq{F}{G}{q}$ and $\attacksq{G}{H}{q}$.
By Lemma~\ref{lem:A2},
$\attacksq{F}{H}{q}$ or $\attacksq{G}{F}{q}$.
In both cases, the attack graph contains a nonterminal cycle, a contradiction.
\end{proof}

\begin{proof}[Lemma~\ref{lem:sharedvariables}]
\framebox{1.}
Let $\tau$ be a join tree for $q$.
By Lemma~\ref{lem:terminalcycles}, every cycle in $q$'s attack graph is of the form $H\attacksymbolq{q}H'\attacksymbolq{q}H$.
We show that if $H\attacksymbolq{q}H'\attacksymbolq{q}H$,
then $\tau$ contains an edge $\{H,H'\}$.
Assume towards a contradiction that $H\attacksymbolq{q}H'\attacksymbolq{q}H$ and there exists $I\in q$ such that $H\neq I\neq H'$ and $I$ lies on the (unique) path in $\tau$ between $H$ and $H'$.
From $\attacksq{H}{H'}{q}$, it follows $\attacksq{H}{I}{q}$.
Then the cycle $H\attacksymbolq{q}H'\attacksymbolq{q}H$ is not terminal, a contradiction.

Assume variable $x$ occurs in two distinct cycles of $q$'s attack graph.
We can assume disjoint cycles $F\attacksymbolq{q}F'\attacksymbolq{q}F$
and $G\attacksymbolq{q}G'\attacksymbolq{q}G$ such that $x\in\all{F}\cup\all{F'}$ and
$x\in\all{G}$.
Assume towards a contradiction $x\not\in\key{F}$.
By the {\em Connectedness Condition\/}, if $e$ is an edge on the unique path in $\tau$ between $F$ and $G$ and $e\neq\{F,F'\}$, then the edge label of $e$ contains $x$.
Since the cycle $F\attacksymbolq{q}F'\attacksymbolq{q}F$ is terminal, $\nattacksq{F}{G}{q}$.
It follows $x\in\keycl{F}{q}$.
Since $x\not\in\key{F}$, we can assume an atom $H\in q$ such that $H\neq F$ and $\key{H}\subseteq\key{F}$.
From $\attacksq{F}{F'}{q}$ and Lemma~\ref{lem:rsu}, it follows $H\neq F'$.
By the premise of the lemma, we can assume $H'\in q\setminus\{F,F'\}$ such that $H\attacksymbolq{q}H'\attacksymbolq{q}H$.
By the {\em Connectedness Condition\/}, every edge label on the path in $\tau$ between $F$ and $H$ contains $\key{H}$.
Let $H\step{L_{1}}I_{1}\step{L_{2}}I_{2}\dots\step{L_{m-1}}I_{m-1}\step{L_{m}}I_{m}$ with $I_{m}=F$ be the unique path in $\tau$ between $H$ and $F$.
Two cases can occur.
\begin{description}
\item[Case $I_{1}=H'$.]
From $\nattacksq{H'}{I_{2}}{q}$, it follows $L_{2}\subseteq\keycl{H'}{q}$.
Since $\key{H}\subseteq L_{2}$, $\key{H}\subseteq\keycl{H'}{q}$.
\item[Case $I_{1}\neq H'$.]
Since $\attacksq{H'}{H}{q}$ and $\nattacksq{H'}{I_{1}}{q}$, it follows $L_{1}\subseteq\keycl{H'}{q}$.
Since $\key{H}\subseteq L_{1}$, $\key{H}\subseteq\keycl{H'}{q}$.
\end{description}
Hence, $\key{H}\subseteq\keycl{H'}{q}$.
Then by Lemma~\ref{lem:rsu}, $\nattacksq{H'}{H}{q}$, a contradiction.
We conclude by contradiction $x\in\key{F}$.

\framebox{2.}
Assume $\attacks{F}{G}{q}$ is a weak attack.
Let $x\in\key{G}\setminus\key{F}$.
By property 1 proved above, for all $H\in q\setminus\{F,G\}$, $x\not\in\all{H}$.  
Since the attack $\attacks{F}{G}{q}$ is weak, $\FD{q}\models\fd{\key{F}}{x}$.
Since $x$ does not occur in $q\setminus\{F,G\}$, it must be the case that $x\in\all{F}$.
It follows $\key{G}\subseteq\all{F}$.
\end{proof}

\begin{proof}[Lemma~\ref{lem:SE3}]
Assume that $F$ is an $R$-atom.
Since $\db$ is purified relative to $q$,
we can assume that the set of $R$-facts of $\db$ is $\calR=\{\substitute{F}{\vec{y}}{\vec{b}_{1}},\dots$, $\substitute{F}{\vec{y}}{\vec{b}_{\ell}}\}$ for some $\ell\geq 0$ and $\vec{b}_{1},\dots,\vec{b}_{\ell}\in D^{\card{\vec{y}}}$.
Since  $\key{F}=\emptyset$, all facts of $\calR$ are key-equal.

\framebox{1$\implies$2}
Since $\db\models q$, we have $\ell\geq 1$, hence $\db\neq\emptyset$.
Let $\rep$ be a repair of $\db$ and $i\in\{1,\dots,\ell\}$.
We need to show that $\rep\models\substitute{q'}{\vec{y}}{\vec{b}_{i}}$.
Since $\formula{\rep\setminus\calR}\cup\{\substitute{F}{\vec{y}}{\vec{b}_{i}}\}$ is a repair of $\db$,
it follows from the premise that $\formula{\rep\setminus\calR}\cup\{\substitute{F}{\vec{y}}{\vec{b}_{i}}\}\models q$. 
Since $q$ contains no self-join, it follows $\formula{\rep\setminus\calR}\models\substitute{q'}{\vec{y}}{\vec{b}_{i}}$,
hence $\rep\models\substitute{q'}{\vec{y}}{\vec{b}_{i}}$.

\framebox{2$\implies$1}
Let $\rep$ be a repair of $\db$.
By the premise, for every $j\in\{1,\dots,\ell\}$, $\rep\models\substitute{q'}{\vec{y}}{\vec{b}_{j}}$.
We can assume $i\in\{1,\dots,\ell\}$ such that $\substitute{F}{\vec{y}}{\vec{b}_{i}}\in\rep$.
From $\rep\models\substitute{q'}{\vec{y}}{\vec{b}_{i}}$, it follows $\rep\models q$.
\end{proof}

\begin{subproof}[Sublemma~\ref{prop:shared}]
\framebox{1$\implies$2}
Let $\rep=\rep_{1}\cup\rep_{2}$ be a repair of $\db$ such that for all $i\in\{1,\dots,\ell\}$,
for all partitions $P$ of $\db_{i}$,
\begin{itemize}
\item 
if $P\not\in\cqak{q_{i}}$, then $\rep_{1}$ contains a repair of $P$ falsifying $q_{i}$; and
\item
if $P\in\cqak{q_{i}}$, then $\rep_{2}$ contains a repair of $P$.
\end{itemize}
Since $\db\in\cqak{q}$ by the premise, we have $\rep\models q$.
For all $i\in\{1,\dots,\ell\}$, for every valuation $\theta$, if $\theta(q_{i})\subseteq\rep$,
then $\theta(F_{i}),\theta(G_{i})$ must belong to the same partition of $\db_{i}$,
hence $\theta(F_{i}),\theta(G_{i})\subseteq\rep_{2}$. 
Consequently, $\rep_{2}\models q$.
Since $\rep_{2}\subseteq\bigcup_{1\leq i\leq\ell}\clean{\db_{i}}$,
we have $\bigcup_{1\leq i\leq\ell}\clean{\db_{i}}\models q$. 

\framebox{2$\implies$1}
Let $\rep$ be a repair of $\db$.
Let $\vec{x}$ be a sequence of distinct variables containing every variable $x$ such that for some $1\leq i<j\leq\ell$, $x\in\all{q_{i}}\cap\all{q_{j}}$.
By the premise, we can assume $\vec{a}\in D^{\card{\vec{x}}}$ such that $\bigcup_{1\leq i\leq\ell}\clean{\db_{i}}\models\substitute{q}{\vec{x}}{\vec{a}}$.
Let $\theta$ be the valuation over $\sequencevars{\vec{x}}$ such that $\theta(\vec{x})=\vec{a}$.
For $1\leq i\leq\ell$, $\sequencevars{\vec{x}_{i}}\subseteq\sequencevars{\vec{x}}$.
For $1\leq i\leq\ell$, define $\vec{a}_{i}\in D^{\card{\vec{x}_{i}}}$ as the sequence of constants such that $\vec{a}_{i}=\theta(\vec{x}_{i})$.
Then, for each $i\in\{1,\dots,\ell\}$, $\clean{\db_{i}}\models\substitute{q_{i}}{\vec{x}_{i}}{\vec{a}_{i}}$.
Since $\clean{\db_{i}}$ contains a partition with vector $\vec{a}_{i}$ which belongs to $\cqak{q_{i}}$ (by construction),
every repair of $\clean{\db_{i}}$ satisfies $\substitute{q_{i}}{\vec{x}_{i}}{\vec{a}_{i}}$.
Since $\rep$ contains a repair of $\clean{\db_{i}}$, we conclude $\rep\models\substitute{q_{i}}{\vec{x}_{i}}{\vec{a}_{i}}$.
Since $\vec{x}$ contains all variables that occur in two distinct queries among $q_{1},\dots,q_{\ell}$, 
it is correct to conclude $\rep\models\substitute{q}{\vec{x}}{\vec{a}}$.
Since $\rep$ is an arbitrary repair of $\db$, $\db\in\cqak{q}$.
\end{subproof}

\begin{proof}[Lemma~\ref{lem:cq}]
Let $\calD$ be the set of uncertain databases.
We define a mapping $f:\calD\longrightarrow\calD$ as follows.
If $\db$ is an uncertain database with active domain $D$, then $f(\db)$ is the smallest set such that 
\begin{itemize}
\item
for every fact $A\in\db$, if the relation name of $A$ occurs in $q'$, then $A\in f(\db)$; and
\item
whenever $R(\underline{\vec{x}})$ in $q\setminus q'$ and $\vec{a}\in D^{\card{\vec{x}}}$,
then $f(\db)$ contains $R(\underline{\vec{a}})$.
\end{itemize}
It is straightforward that $f$ is first-order expressible and $\db\in\cqak{q'}\iff f(\db)\in\cqak{q}$.
\end{proof}

\begin{proof}[Corollary~\ref{cor:fuxman}]
By Lemma~\ref{lem:cq}, there exists an \ACO\ many-one reduction from $\cqak{\cq{k}}$ to $\cqak{\acq{k}}$.
Then by Theorem~\ref{the:acq}, $\cqak{\cq{k}}$ is in \PTIME.
\end{proof}

\section{Proofs of Section~\ref{sec:bid}}

\begin{proof}[Proposition~\ref{prop:bid}]
\framebox{1$\implies$2}
Let $\pw\in\pworlds{\db}$ such that $\prob{\pw}>0$.
We need to show $\pw\models q$.
From $\prob{\pw}>0$, it follows that for every block $\block$ of $\db$,
if $\sum_{A\in\block}\prob{A}=1$, then $\pw$ contains a fact of $\block$.
Consequently, there exists a repair $\rep$ of $\db'$ such that $\rep\subseteq\pw$.
Since $\rep\models q$ by the premise, it follows $\pw\models q$.

\framebox{2$\implies$1}
Let $\rep$ be a repair of $\db'$, hence $\rep\in\pworlds{\db}$.
We have $\prob{\rep}>0$, because for every block $\block$ of $\db$,
if $\sum_{A\in\block}\prob{A}=1$, then $\rep$ contains a fact of $\block$.
By the premise, $\rep\models q$.
\end{proof}

\begin{proof}[Corollary~\ref{cor:prob}]
Assume $\cqak{q}$ is not first-order expressible.
By Theorem~\ref{the:prob}, query $q$ is  not safe.
By Theorem~\ref{the:suciu}, $\prok{q}$ is \sharpP-hard.
\end{proof}

\end{document}